%% file: ContextualityAndBundleDiagrams.tex
\documentclass[prl,reprint]{revtex4-1}
\usepackage{braket}
\usepackage{mathtools}
\usepackage{adjustbox}
\usepackage[utf8]{inputenc}

\usepackage[english]{babel}
\usepackage{csquotes}
\usepackage{capt-of}
\usepackage[dvipsnames]{xcolor}
\usepackage{color}
\usepackage{tikz}
\usepackage{graphics}
\usepackage{float}
\usetikzlibrary{shapes.geometric}
\usetikzlibrary{backgrounds}
\usepackage{amsmath}
\usepackage{amssymb}
\usepackage{amsthm}
\usepackage{braket}
\usepackage{bbm}

\theoremstyle{definition}
\newtheorem{defi}{Definition}

\theoremstyle{plain}
\newtheorem{lem}[defi]{Lemma}
\newtheorem{prop}[defi]{Proposition}

\usepackage{amsmath, amssymb} 

\DeclareMathOperator{\res}{res}
\DeclareMathOperator{\id}{id}

\usepackage[draft]{todonotes} 

\begin{document}
\input{1_Introduction}

\input{2_Sheafs}

\input{3_Computation}

\input{4_ClusterRing}

\input{5_Conclusion}
\bibliographystyle{apsrev4-1}
\bibliography{paper}
\end{document}

%% file: 1_Introduction.tex
\begin{abstract} 
Contextuality describes the nontrivial dependence of measurement outcomes on particular choices of jointly measurable observables. In this work we review and generalize the \emph{bundle diagram representation} introduced in [S.\ Abramsky \emph{et al.}, 24th EACSL Annual Conference on Computer Science Logic, \textbf{41}, 211--228, (2015)]
in order to graphically demonstrate the contextuality of diverse empirical models.  
\end{abstract}

\title{Contextuality and bundle diagrams} 
\author{Kerstin Beer}\author{Tobias J.\ Osborne}
\affiliation{Institut f\"ur Theoretische Physik, Leibniz Universit\"at Hannover, Germany}
\date{\today}
\maketitle

\section{Introduction}
Contextuality -- the general impossibility of assigning \emph{predetermined} outcomes to measurements of quantum observables -- is a crucial feature of quantum mechanics \cite{Kochen1975,Bell1964}. Quantum contextuality has a wide variety of applications from understanding Bell nonlocality through to device-independent information processing \cite{Heywood1983, Stairs1983, Mermin1990, Raussendorf2012, Raussendorf2000, Nielsen2006, Gross2007}. Recently, it has been realised that contextuality is the fundamental resource enabling quantum computational processes \cite{Anders2009,Delfosse2015,Howard2014,Raussendorf2015}. There has been rapid recent progress in understanding contextuality, both theoretically {\cite{Anders2009,Cabello2014,Raussendorf2016,Okay2017,Abramsky2012a,Abramsky2011} and experimentally \cite{Amselem2009,Bartosik2009,Cabello2008a,Cabello2008,Kirchmair2009,Lapkiewicz2011,Michler2000,Spekkens2009,Yu2012,Zhang2013}, leading to many insights into the behaviour of contextual models. However, many mysteries remain and a general theory of contextuality analogous to the theory of quantum entanglement, remains to be achieved \cite{Amaral2015,Gnacinski2016,Kurzynski2017}.

Several frameworks have been developed to study quantum contextuality exploiting, variously, graph and hypergraph theory, sheaf theory, and cohomology theory  \cite{Cabello2014,Raussendorf2016,Okay2017,Raussendorf2016a,Stephen2016}. Most relevant to this work are a sequence of works by Abramsky and coworkers aiming to quantify contextuality using the technology of \emph{sheaf theory}. Here one conceives of an \emph{empirical model} -- a collection of probability distributions for jointly measurable observables -- as a kind of ``vector bundle'': the observables of the empirical model are points in the ``base space'' where two points are connected if they may be jointly measured (i.e., if they are in the same context). The possible measurement outcomes are then associated to the contexts (the edges) similar to transition maps. This analogy is not just superficial as one can exploit the theory of sheaves, developed to generalise that of vector bundles, to exactly capture the data of an empirical model. Exploiting results from sheaf theory one can formulate various obstructions to non-contextuality in terms of properties of the sections of the sheaf.  

The formalism of sheaf theory is very powerful and many key results have been obtained via its application. However, it is challenging for the newcomer to appreciate its utility. In particular, although it is motivated by vector bundle theory, a geometrical interpretation of the sheaf-theoretic formulation of an empirical model may be unclear. To clarify its geometrical aspects  Abramsky \emph{et al.}\ \cite{Abramsky2015} introduced a graphical method via \emph{bundle diagrams} to visualise empirical models, particularly in the setting where there are two agents who can each choose to measure one of two possible observables. Bundle diagrams are visually striking and of considerable utility in understanding contextuality. For this reason we were interested in generalising the bundle diagram approach as a pedagogical tool to study empirical models involving more agents and measurement settings. 

In this paper we study a variety of empirical models going beyond the two-agent scenario and explain how to adapt the bundle diagram technology of Abramsky \emph{et al.}\ \cite{Abramsky2015} and Car\`u \cite{Caru2017} to illustrate the contextuality of these empirical models. 
In the first section we review the sheaf-theoretic formulation of empirical models and therefore recall the most important definitions from \cite{Abramsky2011}. Introducing the bundle diagram representation helps one to geometrically understand the sheaf structure of the empirical models. Furthermore we review how one can deduce contextuality from the non-existence of a global section. In the following we represent a measurement-based quantum computational model via a bundle diagram. We then graphically illustrate that, for a Greenberger-Horne-Zeilinger type scenario, there exists no empirical model with a global section. We conclude with a discussion of ane examples arising from the cluster state on a ring.

%% file: 2_Sheafs.tex
\section{The sheaf-theoretic structure of empirical models}

The aim of this section is to understand how to associate to an empirical model a sheaf and to describe contextuality via this mathematical structure. We restrict ourselves to the case where Pauli measurement operations are the observables. Throughout this section we very closely follow the  papers of Abramsky and coworkers \cite{Abramsky2011,Abramsky2012a,Abramsky2015,Caru2017}.

In the basic setting there are several agents, who can each select from a set of measurements and observe outcomes. We call the procedure whereby each agent performs a measurement on their system and observes an outcome, an \emph{event}. A probability distribution on these events results from repeated trials. 
\begin{defi}
	An \emph{empirical model} is a family of probability distributions on events, one for each choice of measurements. A set of allowed jointly measurable observables is called a \emph{measurement context}.\label{firstdefi}
	
\end{defi}
We exemplify the theory in terms of a bipartite qubit model, where each of the two parties, Alice and Bob, can apply a Pauli measurement operation to the state
\begin{align*}
	\ket{\Psi}&=\frac{1}{\sqrt{2}}(\ket{0}_A\otimes\ket{0}_B+\ket{1}_A\otimes\ket{1}_B) \equiv\frac{1}{\sqrt{2}}(\ket{00}+\ket{11})\text{.} 
\end{align*}
We choose for the allowed observables $X=\ket{0}\bra{1}+\ket{1}\bra{0}$ and $Z=\ket{0}\bra{0}-\ket{1}\bra{1}$ and index the measurement operations according to whether Alice or Bob carries out the measurement. The measurement setting where Alice measures $X_A=X\otimes\mathbbm{1}$ and Bob $X_B=\mathbbm{1}\otimes X$ is an example of a pair of jointly measurable observables (the two observables commute). For this bipartite system we take for our contexts the following four sets of jointly measurable observables,
\begin{align*}
	C_1&=\{X_A,X_B\},\\C_2&=\{X_A,Z_B\},\\C_3&=\{Z_A,X_B\},\quad \text{and}\\C_4&=\{Z_A,Z_B\}\text{.}
\end{align*}

Although the measurements applied depend on the context, the outcomes of the measurements can be described in a context-free manner because the eigenvalues of both $X$ and $Z$ are $\pm1$. To this end we label the outcome $(-1)^{j}$ where $j\in\{0,1\}$ is a bit indicating the outcome.  Thus, the probabilities that the measurement outcome is $j$ for Alice and $k$ for Bob may be summarised in the following table.
\begin{center}
	\begin{tabular}{c c | c c c c}
		\hline
		$A$ & $B$ & $ 0 0 $ & $ 1 0 $ & $ 0 1 $ & $ 1 1 $\\
		\hline
		$X_A$ & $X_B$ & $1/2$ & $0$ & $0$ & $1/2$\\
		$X_A$ & $Z_B$ & $1/4$ &   $1/4$ &   $1/4$ &  $1/4$\\
		$Z_A$ & $X_B$ & $1/4$ &   $1/4$ &   $1/4$ &  $1/4$\\
		$Z_A$ & $Z_B$ & $1/2$ & $0$ & $0$ & $1/2$\\
	\end{tabular}
	\captionof{table}{Empirical model for $\ket{\Psi}$.}\label{empiricalmodel}
\end{center}
\subsection{Joint measurability structures and abstract simplicial complexes}
In  Abramsky \emph{et al.}\ \cite{Abramsky2015} introduced a diagrammatic representation -- further developed in \cite{Caru2017} -- to depict empirical models whereby measurements are represented as vertices in a ``base space'' and possible outcomes as ``fibers'' above the base.

To describe this representation we first note \cite{Henson2012,Kunjwal2014} that a contextual model, or \emph{joint measurability structure}, may be represented via a combinatorial object known as an \emph{abstract simplicial complex}. Recall that a POVM is a map $A:\Omega\rightarrow \mathcal{B}^+(\mathcal{H})$ from an \emph{outcome set} (which we take to be finite from now on) to the convex cone of positive operators on a Hilbert space $\mathcal{H}$ such that 
\begin{equation}
	\sum_{j\in \Omega} A(j) = \mathbb{I}.
\end{equation}
We say that a POVM $A$ with outcome set $\Omega_1\times \cdots\times \Omega_n$ marginalises to a set of POVMs $\{A_1, \ldots, A_n\}$ with outcome sets $\{\Omega_1,\ldots, \Omega_n\}$, respectively, if 
\begin{equation}
	\sum_{j_1,\ldots, \widehat{j}_k, \ldots, j_n} A(j_1,\ldots, j_n) = A_k(j_k)
\end{equation}
for all $k=1,2,\ldots, n$, where the hat means that the variable is excluded from the summation. If, for a given set $\{A_1, \ldots, A_n\}$ of POVMs there is such a measurement $A$ marginalising to $\{A_1, \ldots, A_n\}$ then we say they are \emph{jointly measurable}. Thus, if a set of POVMs is jointly measurable then so is any subset of them. 

An elegant combinatorial object which naturally captures the structure of a contextual model is that of an abstract simplicial complex. 
\begin{defi}
	A family $\Delta$ of nonempty subsets of a set $M$ is an \emph{abstract simplicial complex} if, for every set $U\in \Delta$ and every nonempty subset $V\subset U$, $V$ also belongs to $\Delta$. The finite sets belonging to $\Delta$ are called \emph{faces} and the \emph{vertices} are the elements of the set $\bigcup\Delta$. (We henceforth assume that $M=\bigcup \Delta$.)
\end{defi}
An abstract simplicial complex $\Delta$ gives rise to a topological space by endowing it with the \emph{Alexandroff topology} by defining a subset $\mathcal{U} \subset \Delta$ to be \emph{closed} if and only if $\mathcal{U}$ is itself an abstract simplicial complex, i.e., $\forall U\in \mathcal{U}$ if $V\subset U$ then $V\in \mathcal{U}$. The open sets are hence generated by \emph{stars}, where for a face $\sigma\in \Delta$ we set $\text{star}(\sigma) \equiv \{\tau \in \Delta\,|\,\sigma\in\tau\}$. In particular, this means that \emph{upper sets} in the poset of faces with respect to inclusion, i.e.\ maximal faces, are open. 

We henceforth take for the vertices of the abstract simplicial complex associated to a contextual model the set of all  allowed observables: a set of vertices form a face whenever the corresponding measurements can be performed jointly. Thus contexts correspond to faces of a such complex. This complex is called the \emph{base}.

Continuing the two-qubit example above, the base suitable to the above described model is the complex $\{ \{X_A\}, \{X_B\}, \{Z_A\}, \{Z_B\}, C_1, C_2, C_3, C_4\}$. This may be recognised as a square, where the four vertices represent the observables and the edges the contexts:
\begin{figure}[H]
	\centering
	\begin{tikzpicture}[scale=2.0]
	\foreach \u in {(0,0),(1,0),(1,-1),(0,-1)}
	\node[fill,circle,scale=0.5] at \u{};
	\coordinate (XA) at (0,0);
	\node at (XA) [left = 1mm of XA] {$X_A$};
	\coordinate (XB) at (1,0);
	\node at (XB) [right = 1mm of XB] {$X_B$};
	\coordinate (ZA) at (1,-1);
	\node at (ZA) [right = 1mm of ZA] {$Z_A$};
	\coordinate (ZB) at (0,-1);
	\node at (ZB) [left = 1mm of ZB] {$Z_B$};
	\draw[line width=1pt] (XA) 	-- (XB);
	\draw[line width=1pt] (XB) 	-- (ZA);
	\draw[line width=1pt] (ZA)	-- (ZB);
	\draw[line width=1pt] (ZB) 	-- (XA);
	\end{tikzpicture}
	\captionof{figure}{Base of the bundle diagram.}
\end{figure}
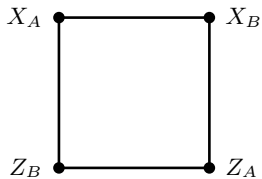

\subsection{Bundle diagrams} 
We can enhance the abstract simplicial complex representation by including information about the \emph{possible} measurement outcomes of the observables of a context. To this end we attach, above each vertex, the \emph{fiber} (or, more correctly, the \emph{stalk}) of possible measurement outcomes, which in this example are either $0$ or $1$: 
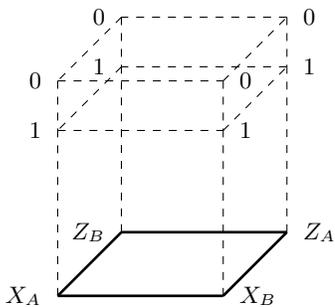
\begin{figure}[H]
	\centering
	\begin{tikzpicture}[scale=2.2]
	\coordinate (XA) at (0,0,0);
	\node at (XA) [left = 1mm of XA] {$X_A$};
	\coordinate (XB) at (1,0,0);
	\node at (XB) [right = 1mm of XB] {$X_B$};
	\coordinate (ZA) at (1,0,-1);
	\node at (ZA) [right = 1mm of ZA] {$Z_A$};
	\coordinate (ZB) at (0,0,-1);
	\node at (ZB) [left = 1mm of ZB] {$Z_B$};
	\draw[line width=1pt] (XA) 	-- (XB);
	\draw[line width=1pt] (XB) 	-- (ZA);
	\draw[line width=1pt] (ZA)	-- (ZB);
	\draw[line width=1pt] (ZB) 	-- (XA);
	\coordinate (1XA) at (0,1,0);
	\node at (1XA) [left = 1mm of 1XA] {$1$};
	\coordinate (1XB) at (1,1,0);
	\node at (1XB) [right = 1mm of 1XB] {$1$};
	\coordinate (1ZA) at (1,1,-1);
	\node at (1ZA) [right = 1mm of 1ZA] {$1$};
	\coordinate (1ZB) at (0,1,-1);
	\node at (1ZB) [left = 1mm of 1ZB] {$1$};
	\draw[dashed]  (XA) 	-- (1XA);
	\draw[dashed]  (XB) 	-- (1XB);
	\draw[dashed]  (ZA)	-- (1ZA);
	\draw[dashed]  (ZB) 	-- (1ZB);
	\draw[dashed] (1XA) 	-- (1XB);
	\draw[dashed] (1XB) 	-- (1ZA);
	\draw[dashed] (1ZB) 	-- (1XA);
	\draw[dashed] (1ZA) 	-- (1ZB);
	\coordinate (0XA) at (0,1.3,0);
	\node at (0XA) [left = 1mm of 0XA] {$0$};
	\coordinate (0XB) at (1,1.3,0);
	\node at (0XB) [right = 1mm of 0XB] {$0$};
	\coordinate (0ZA) at (1,1.3,-1);
	\node at (0ZA) [right = 1mm of 0ZA] {$0$};
	\coordinate (0ZB) at (0,1.3,-1);
	\node at (0ZB) [left = 1mm of 0ZB] {$0$};
	\draw[dashed]  (0XA) 	-- (1XA);
	\draw[dashed]  (0XB) 	-- (1XB);
	\draw[dashed]  (0ZA)	-- (1ZA);
	\draw[dashed]  (0ZB) 	-- (1ZB);
	\draw[dashed] (0XA) 	-- (0XB);
	\draw[dashed] (0XB) 	-- (0ZA);
	\draw[dashed] (0ZB) 	-- (0XA);
	\draw[dashed] (0ZA) 	-- (0ZB);
	\end{tikzpicture}
	\captionof{figure}{Fibers of the bundle diagram.}
\end{figure}

Depending on the empirical model not all measurement outcomes can occur. We illustrate this by connecting elements of the fibers above a face (i.e., above a context) if the corresponding outcome combination has a probability strictly larger than $0$. For the two-qubit example we depict this by connecting possible outcomes with purple lines:
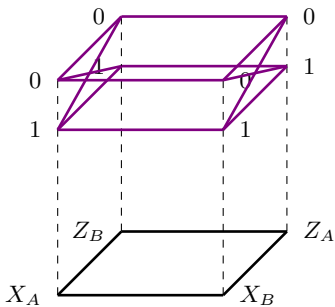
\begin{figure}[H]
	\centering
		\begin{tikzpicture}[scale=2.2]
		\coordinate (XA) at (0,0,0);
		\node at (XA) [left = 1mm of XA] {$X_A$};
		\coordinate (XB) at (1,0,0);
		\node at (XB) [right = 1mm of XB] {$X_B$};
		\coordinate (ZA) at (1,0,-1);
		\node at (ZA) [right = 1mm of ZA] {$Z_A$};
		\coordinate (ZB) at (0,0,-1);
		\node at (ZB) [left = 1mm of ZB] {$Z_B$};
		\draw[line width=1pt] (XA) 	-- (XB);
		\draw[line width=1pt] (XB) 	-- (ZA);
		\draw[line width=1pt] (ZA)	-- (ZB);
		\draw[line width=1pt] (ZB) 	-- (XA);
		\coordinate (1XA) at (0,1,0);
		\node at (1XA) [left = 1mm of 1XA] {$1$};
		\coordinate (1XB) at (1,1,0);
		\node at (1XB) [right = 1mm of 1XB] {$1$};
		\coordinate (1ZA) at (1,1,-1);
		\node at (1ZA) [right = 1mm of 1ZA] {$1$};
		\coordinate (1ZB) at (0,1,-1);
		\node at (1ZB) [left = 1mm of 1ZB] {$1$};
		\draw[dashed]  (XA) 	-- (1XA);
		\draw[dashed]  (XB) 	-- (1XB);
		\draw[dashed]  (ZA)	-- (1ZA);
		\draw[dashed]  (ZB) 	-- (1ZB);
		\draw[violet,line width=1pt] (1XA) 	-- (1XB);
		\draw[violet,line width=1pt] (1XB) 	-- (1ZA);
		\draw[violet,line width=1pt] (1ZB) 	-- (1XA);
		\draw[violet,line width=1pt] (1ZA) 	-- (1ZB);
		\coordinate (0XA) at (0,1.3,0);
		\node at (0XA) [left = 1mm of 0XA] {$0$};
		\coordinate (0XB) at (1,1.3,0);
		\node at (0XB) [right = 1mm of 0XB] {$0$};
		\coordinate (0ZA) at (1,1.3,-1);
		\node at (0ZA) [right = 1mm of 0ZA] {$0$};
		\coordinate (0ZB) at (0,1.3,-1);
		\node at (0ZB) [left = 1mm of 0ZB] {$0$};
		\draw[dashed]  (0XA) 	-- (1XA);
		\draw[dashed]  (0XB) 	-- (1XB);
		\draw[dashed]  (0ZA)	-- (1ZA);
		\draw[dashed]  (0ZB) 	-- (1ZB);
		\draw[violet,line width=1pt] (0XA) 	-- (0XB);
		\draw[violet,line width=1pt] (0XB) 	-- (0ZA);
		\draw[violet,line width=1pt] (0ZB) 	-- (0XA);
		\draw[violet,line width=1pt] (0ZA) 	-- (0ZB);
		\draw[violet,line width=1pt] (0XA) 	-- (1ZB);
		\draw[violet,line width=1pt] (1XA) 	-- (0ZB);
		\draw[violet,line width=1pt] (0XB) 	-- (1ZA);
		\draw[violet,line width=1pt] (1XB) 	-- (0ZA);
		\end{tikzpicture}
	\captionof{figure}{Bundle diagram for the empirical model of Table~\ref{empiricalmodel}.}\label{qubitbundle}
\end{figure}
The example illustrated above is noncontextual. Before proceeding, it is instructive to give an example of what a contextual empirical model looks like. To this end we consider the \emph{Popescu-Rohrlich box} \cite{Popescu1994}, which corresponds to a joint measurement scenario of a particular extremal state of a generalised probabilistic theory whose correlations exceed those allowed by quantum mechanics. Measurements of this state give rise to the following empirical model.
\begin{figure}[H]
	\centering
	\begin{tabular}{c c | c c c c}
		\hline
		$A$ & $B$ & $ 0 0 $ & $ 1 0 $ & $ 0 1 $ & $ 1 1 $\\
		\hline
		$N_A$ & $N_B$ & $1/2$ & $0$ & $0$ & $1/2$\\
		$N_A$ & $M_B$ & $1/2$ & $0$ & $0$ & $1/2$\\
		$M_A$ & $N_B$ & $1/2$ & $0$ & $0$ & $1/2$\\
		$M_A$ & $M_B$ & $0$ & $1/2$ & $1/2$ & $0$\\
	\end{tabular}
	\captionof{table}{Empirical model for the Popescu-Rohlich box.}\label{empiricalmodel3}
\end{figure}
The corresponding bundle diagram is illustrated here:
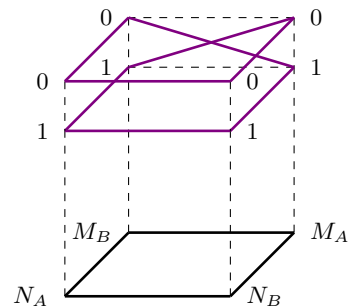
\begin{figure}[H]
	\centering
		\begin{tikzpicture}[scale=2.2]
		\coordinate (XA) at (0,0,0);
		\node at (XA) [left = 1mm of XA] {$N_A$};
		\coordinate (XB) at (1,0,0);
		\node at (XB) [right = 1mm of XB] {$N_B$};
		\coordinate (ZA) at (1,0,-1);
		\node at (ZA) [right = 1mm of ZA] {$M_A$};
		\coordinate (ZB) at (0,0,-1);
		\node at (ZB) [left = 1mm of ZB] {$M_B$};
		\draw[line width=1pt] (XA) 	-- (XB);
		\draw[line width=1pt] (XB) 	-- (ZA);
		\draw[line width=1pt] (ZA)	-- (ZB);
		\draw[line width=1pt] (ZB) 	-- (XA);
		\coordinate (1XA) at (0,1,0);
		\node at (1XA) [left = 1mm of 1XA] {$1$};
		\coordinate (1XB) at (1,1,0);
		\node at (1XB) [right = 1mm of 1XB] {$1$};
		\coordinate (1ZA) at (1,1,-1);
		\node at (1ZA) [right = 1mm of 1ZA] {$1$};
		\coordinate (1ZB) at (0,1,-1);
		\node at (1ZB) [left = 1mm of 1ZB] {$1$};
		\draw[dashed]  (XA) 	-- (1XA);
		\draw[dashed]  (XB) 	-- (1XB);
		\draw[dashed]  (ZA)	-- (1ZA);
		\draw[dashed]  (ZB) 	-- (1ZB);
		\draw[violet,line width=1pt] (1XA) 	-- (1XB);
		\draw[violet,line width=1pt] (1XB) 	-- (1ZA);
		\draw[violet,line width=1pt] (1ZB) 	-- (1XA);
		\draw[dashed] (1ZA) 	-- (1ZB);
		\coordinate (0XA) at (0,1.3,0);
		\node at (0XA) [left = 1mm of 0XA] {$0$};
		\coordinate (0XB) at (1,1.3,0);
		\node at (0XB) [right = 1mm of 0XB] {$0$};
		\coordinate (0ZA) at (1,1.3,-1);
		\node at (0ZA) [right = 1mm of 0ZA] {$0$};
		\coordinate (0ZB) at (0,1.3,-1);
		\node at (0ZB) [left = 1mm of 0ZB] {$0$};
		\draw[dashed]  (0XA) 	-- (1XA);
		\draw[dashed]  (0XB) 	-- (1XB);
		\draw[dashed]  (0ZA)	-- (1ZA);
		\draw[dashed]  (0ZB) 	-- (1ZB);
		\draw[violet,line width=1pt] (0XA) 	-- (0XB);
		\draw[violet,line width=1pt] (0XB) 	-- (0ZA);
		\draw[violet,line width=1pt] (0ZB) 	-- (0XA);
		\draw[dashed] (0ZA) 	-- (0ZB);
		\draw[violet,line width=1pt] (0ZA) 	-- (1ZB);
		\draw[violet,line width=1pt] (1ZA) 	-- (0ZB);
		\end{tikzpicture}
	\captionof{figure}{Bundle diagram for the empirical model of Table~\ref{empiricalmodel3}.}\label{qubitbundle3}
\end{figure}
This model is contextual and there is no local hidden variable model explaining the measurement outcomes.

At this point it is clear that the bundle diagram representation can only describe a \emph{possibilistic model} as it only contains information about what possible outcomes can occur, and not the corresponding probability. Nevertheless, we will see that such bundle diagrams are of great utility for representing the contextuality of an empirical model.  

In order to more precisely describe a useful graphical representation for an empirical model it is helpful to first introduce some of the mathematical terminology of sheaf theory. 

\subsection{Sheaf structure}
We denote by $M$ the vertices of an abstract simplicial complex $\Delta$ corresponding to a contextual model \footnote{It assumed throughout that the set of allowed observables $M$ for the model and the set of outcomes $O$ are finite.}. We give this set the discrete topology by defining any subset of $M$ to be open, i.e., the topology for $M$ is the power set $\mathcal{P}(M)$. (Note: this topological space is very different from the natural topological space $\Delta$ given by the Alexandroff topology.) We first build a structure called a \emph{presheaf of sets} over the space $(M,\mathcal{P}(M))$. This presheaf is intended to capture all possible measurement outcomes, regardless of whether the observables are jointly measurable or not (note that, depending on whether $U$ is a context, there may or may not be a joint measurement for the observables in $U$). We do this by associating to every subset $U\subset M$ the set of all functions from $U$ to $O$. This is a finite set, denoted $O^U$, with cardinality $|O|^{|U|}$. We call a function $s:U\rightarrow O$ a \emph{section over $U$}. A section in $\mathcal{E}(M)$ is called a \emph{global section}.
\begin{defi}
	If $s : U' \rightarrow O$ is a function and $U\subseteq U'$, we write
	$s|_{U} : U\rightarrow O$ for the restriction of $s$ to $U$. We define the \emph{restriction map} 
	\begin{equation*}
		\res_U^{U'} : \mathcal{E}(U')\rightarrow \mathcal{E}(U)
	\end{equation*} 
	via 
	\begin{equation*}
		\res_U^{U'}(s) \equiv s|_{U}.
	\end{equation*}
	The map $\res_U^{U'}$ endows $\mathcal{E}$ with the structure of a \emph{presheaf} as it enjoys the additional properties that 
	$\res_U^{U'}\circ\res_{U'}^{U''}=\res_U^{U''}$ for  $U\subset U'\subset U''$ and $\res_U^{U}=\id_U$.
\end{defi}
This \emph{presheaf $\mathcal{E}$ of events} has two important additional properties. 
\begin{defi}
\begin{enumerate}
	\item Locality: if $\{V_j\}$ is a covering of $U$ and if $s,t\in \mathcal{E}(U)$ are elements such that $s|_{V_j}=t|_{V_j}$ for all $V_j$ then $s=t$.  
	\item Gluing: suppose $\{V_j\}$ is an open covering of $U\subset M$ and we have a family of sections $\{s_j\in \mathcal{E}(V_j)\}$ with the property that for all $j,k$, 
	\begin{equation*}
		s_j|_{V_j \cap V_k}=s_k|_{V_j \cap V_k},
	\end{equation*}
	then there is a unique section $s\in \mathcal{E}(U)$ such that $s|_{V_j}=s_j$ for all $j$.
\end{enumerate}
\end{defi}
\sloppy
Any presheaf obeying locality and gluing is called a \emph{sheaf} which, in this case, we refer to as the \emph{sheaf of events}. The sheaf condition says that we can uniquely glue together compatible local data, where ``compatible'' means that measurement outcomes agree on observables common to contexts. 

Let $U\in \Delta$ be a context: we now exclude events with probabilities strictly equal to $0$ and write $\mathcal{F}(U)$ for the set of sections over the context. For example, each of the purple lines directly above the set of measurements $C_3$ in Figure \ref{qubitbundle} corresponds to a section in the set $\mathcal{E}(C_3)$:
\begin{align*}
		\mathcal{E}(C_3)=\{&s_{C_{3}00}(Z_A)=0, s_{C_{3}00}(X_B)=0;\\ 
		&s_{C_{3}01}(Z_A)=0, s_{C_{3}01}(X_B)=1;\\
		&s_{C_{3}10}(Z_A)=1, s_{C_{3}10}(X_B)=0;\\
		&s_{C_{3}11}(Z_A)=1, s_{C_{3}11}(X_B)=1\}.
\end{align*}
By only including sections with a nonzero probability of occurring we can obtain a subpresheaf $\mathcal{F}$ of $\mathcal{E}$. This is done by associating to sets $U\subset M$ built from the unions of two or more contexts only those sections which are compatible on the overlaps between the constituent contexts. (Abramsky and coworkers \cite{Abramsky2011} describe some additional properties of this subpresheaf, however, we do not need them for this discussion.)

A key role is played throughout this subject by the existence of \emph{global sections} \cite{Abramsky2011} because a global section glues together a compatible family on a presheaf and allocates a predetermined outcome to each measurement. In this case the measurement outcomes depend only on the measurement operator and not the context. If every section $s$ belongs to a compatible family then the model is called \emph{noncontextual}. Such a model is described by a hidden-variable model. If there is a section which does not belong to a compatible family then the model is said to be \emph{logically contextual}. If no section belongs to a compatible family then the model is \emph{strongly contextual}.  

Abramsky and his coworkers \cite{Abramsky2011} determined locality and non-contextuality in terms of the existence of global sections:
\begin{prop}
	The existence of a global section for an empirical model implies the existence of a non-contextual deterministic hidden-variable model which realizes it.\label{thmglobalsection}
\end{prop}
Thus, given an empirical model, we have to check whether we can extend \emph{every} section to a global one, i.e.\ we can realize every possible outcome combination for a context in a local hidden variable model. If this is the case the model is noncontextual. 
\begin{figure}[H]
	\centering
		\begin{tikzpicture}[scale=2.2]
		\coordinate (XA) at (0,0,0);
		\node at (XA) [left = 1mm of XA] {$X_A$};
		\coordinate (XB) at (1,0,0);
		\node at (XB) [right = 1mm of XB] {$X_B$};
		\coordinate (ZA) at (1,0,-1);
		\node at (ZA) [right = 1mm of ZA] {$Z_A$};
		\coordinate (ZB) at (0,0,-1);
		\node at (ZB) [left = 1mm of ZB] {$Z_B$};
		\draw[line width=1pt] (XA) 	-- (XB);
		\draw[line width=1pt] (XB) 	-- (ZA);
		\draw[line width=1pt] (ZA)	-- (ZB);
		\draw[line width=1pt] (ZB) 	-- (XA);
		\coordinate (1XA) at (0,1,0);
		\node at (1XA) [left = 1mm of 1XA] {$1$};
		\coordinate (1XB) at (1,1,0);
		\node at (1XB) [right = 1mm of 1XB] {$1$};
		\coordinate (1ZA) at (1,1,-1);
		\node at (1ZA) [right = 1mm of 1ZA] {$1$};
		\coordinate (1ZB) at (0,1,-1);
		\node at (1ZB) [left = 1mm of 1ZB] {$1$};
		\draw[dashed]  (XA) 	-- (1XA);
		\draw[dashed]  (XB) 	-- (1XB);
		\draw[dashed]  (ZA)	-- (1ZA);
		\draw[dashed]  (ZB) 	-- (1ZB);
		\draw[violet,line width=1pt] (1XA) 	-- (1XB);
		\draw[dashed] (1XB) 	-- (1ZA);
		\draw[dashed] (1ZB) 	-- (1XA);
		\draw[dashed] (1ZA) 	-- (1ZB);
		\coordinate (0XA) at (0,1.3,0);
		\node at (0XA) [left = 1mm of 0XA] {$0$};
		\coordinate (0XB) at (1,1.3,0);
		\node at (0XB) [right = 1mm of 0XB] {$0$};
		\coordinate (0ZA) at (1,1.3,-1);
		\node at (0ZA) [right = 1mm of 0ZA] {$0$};
		\coordinate (0ZB) at (0,1.3,-1);
		\node at (0ZB) [left = 1mm of 0ZB] {$0$};
		\draw[dashed]  (0XA) 	-- (1XA);
		\draw[dashed]  (0XB) 	-- (1XB);
		\draw[dashed]  (0ZA)	-- (1ZA);
		\draw[dashed]  (0ZB) 	-- (1ZB);
		\draw[dashed] (0XA) 	-- (0XB);
		\draw[dashed] (0XB) 	-- (0ZA);
		\draw[dashed] (0ZB) 	-- (0XA);
		\draw[violet,line width=1pt] (0ZA) 	-- (0ZB);
		\draw[violet,line width=1pt] (1XA) 	-- (0ZB);
		\draw[violet,line width=1pt] (1XB) 	-- (0ZA);
		\end{tikzpicture}
	\captionof{figure}{Example for a global section.}\label{globalsectionqubit}
\end{figure}
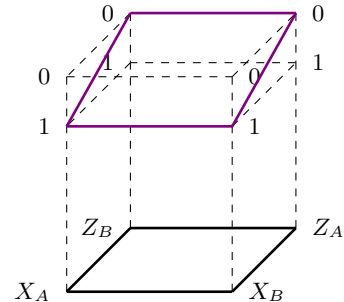
A global section in the bundle diagram is represented by a closed path traversing all the fibers exactly once. This observation allows us to graphically determine the contextuality of an empirical model. 
The above discussed model is noncontextual, because every section is extendable to a global one. An example of a global section is depicted in Figure \ref{globalsectionqubit}.\\

%% file: 3_Computation.tex
\section{Contextuality as a resource for (quantum) computation}
In a key paper Anders and Browne \cite{Anders2009} argued that contextuality is the fundamental resource enabling (quantum) computation. This argument exploits joint measurements of a Greenberger-Horne-Zeilinger-type state $\ket{\Psi}=\frac{1}{\sqrt{2}}(\ket{001}-\ket{110})$ as a resource. By carrying out measurements of this state a classical \emph{parity computer} only able to implement NOT$=\ket{1}\bra{0}+\ket{0}\bra{1}$ and CNOT$=\ket{1}\bra{1} \otimes X + \ket{0}\bra{0} \otimes \mathbbm{1}$ operations may be augmented to carry out any deterministic computation. The canonical gate enabled by measurements of the GHZ state is the NAND-gate: 
\begin{figure}[H]	
	\centering
\begin{tabular}{|c | c| c |}
	\hline 
	a & b & NAND \\ 
	\hline
	0 & 0 & 1 \\ 
	0 & 1 & 1 \\
	1 & 0 & 1 \\
	1 & 1 & 0 \\
	\hline 
\end{tabular}\label{truthtable}
\captionof{table}{Truth table for NAND-gate.}
\end{figure}
By performing three measurements on the GHZ state it is possible to deterministically implement this gate \cite{Browne2006}. Our objective for this section is to explain how to represent the resulting contextual model via a generalised bundle diagram and explore its contextuality in terms of the nonexistence of global sections.

We build a joint measurability structure for the Anders-Browne example as follows. We take for the set of observables $M=\{X_A,X_B,X_C,Y_A,Y_B,Y_C\}$ with  $Y=-i\ket{0}\bra{1}+i\ket{1}\bra{0}$. The choice of which measurements Alice and Bob perform depend on the inputs $a$ and $b$ of the NAND-gate they wish to compute. Charlie measures his observable according to the value of a third supplemented input $a\oplus b$. If the input is $0$ an $X$-measurement is carried out and for $1$ a $Y$-measurement is made. We can summarise this procedure in terms of the following four contexts \begin{align*}
C_1&=\{X_A,X_B,X_C\}, \\C_2&=\{X_A,Y_B,Y_C\}, \\C_3&=\{Y_A,X_B,Y_C\}, \quad\text{and}\\C_4&=\{Y_A,Y_B,X_C\}.
\end{align*} We use the basis
\begin{alignat*}{4}
\ket{+}&=\frac{1}{\sqrt{2}}(\ket{0}+\ket{1}),&\quad\text{and}\quad&
\ket{-}&=\frac{1}{\sqrt{2}}(\ket{0}-\ket{1})\\
\intertext{for the $X$-measurements and}
\ket{\circlearrowleft}&=\frac{1}{\sqrt{2}}(\ket{0}+i\ket{1}),&\quad\text{and}\quad&
\ket{\circlearrowright}&=\frac{1}{\sqrt{2}}(\ket{0}-i\ket{1}), 
\end{alignat*}
for the $Y$-measurements. Measuring these contexts leads to the following empirical model.
\begin{figure}[H]
	\centering
	\begin{tabular}{c |c c c | c c c c c c c c c c }
		&$A$ & $B$ & $C$ & $ 0 0 0 $ & $ 0 0 1 $ & $ 0 1 0 $ & $ 0 1 1 $ & $ 1 0 0 $ & $ 1 0 1 $ & $ 1 1 0 $ & $ 1 1 1 $\\
		\hline
		$C_1$ &$X_A$ & $X_B$ & $X_C$&   $0$ &  $1/4$  &   $1/4$ &   $0$ &  $1/4$  &   $0$ &   $0$ &  $1/4$\\
		$C_2$ &$X_A$ & $Y_B$  &   $Y_C$ &   $0$ &  $1/4$  &   $1/4$ &   $0$ &  $1/4$  &   $0$ &   $0$ &  $1/4$\\
		$C_3$ &$Y_A$ & $X_B$  &   $Y_C$ &   $0$ &  $1/4$  &   $1/4$ &   $0$ &  $1/4$  &   $0$ &   $0$ &  $1/4$\\
		$C_4$ &$Y_A$ & $Y_B$  &   $X_C$ & $1/4$ & $0$ & $0$ & $1/4$ & $0$ & $1/4$ & $1/4$ & $0$
	\end{tabular}
	\captionof{table}{Empirical model for the Anders-Browne example.}
\end{figure}

The abstract simplicial complex $\Delta$ corresponding to this example is depicted in Figure \ref{bottom}.
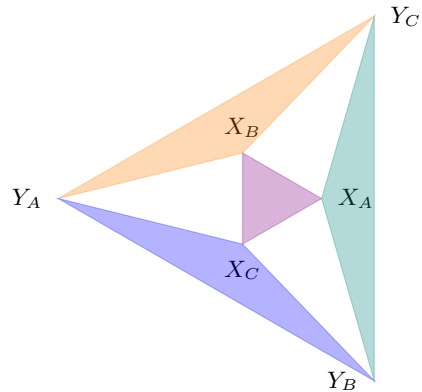
\begin{figure}[H]
	\centering
	\begin{tikzpicture}[scale=1.4]
	\coordinate (XA) at (0.5,0);
	\node at (XA) [right = 1mm of XA] {$X_A$};
	\coordinate (XB) at (-.25,0.43301);
	\node at (XB) [above = 1mm of XB] {$X_B$};
	\coordinate (XC) at (-.25,-0.43301);
	\node at (XC) [below = 1mm of XC] {$X_C$};
	\coordinate (YA) at (-2,0);
	\node at (YA) [left = 1mm of YA] {$Y_A$};
	\coordinate (YB) at (1,-1.7321);
	\node at (YB) [left = 1mm of YB] {$Y_B$};
	\coordinate (YC) at (1,1.7321);
	\node at (YC) [right = 1mm of YC] {$Y_C$};
	\draw[fill,opacity=0.3, violet] (XA) 	-- (XB) --(XC) -- cycle;
	\draw[fill,opacity=0.3, teal] (XA) 	-- (YB) --(YC) -- cycle;
	\draw[fill,opacity=0.3,orange] (YA) 	-- (XB) --(YC) -- cycle;
	\draw[fill,opacity=0.3, blue] (YA) 	-- (YB) --(XC) -- cycle;
	\end{tikzpicture}
	\captionof{figure}{The abstract simplicial complex forming the base for the Anders-Browne example.} \label{bottom}
\end{figure}

\begin{figure}[H]
	\centering
	\begin{tikzpicture}[scale=1.8]
	\coordinate (XA) at (0.5,0,0);
	\node at (XA) [right = 1mm of XA] {$X_A$};
	\coordinate (XB) at (-.25,0,0.43301);
	\node at (XB) [below = 1mm of XB] {$X_B$};
	\coordinate (XC) at (-.25,0,-0.43301);
	\node at (XC) [below = 1mm of XC] {$X_C$};
	\coordinate (YA) at (-2,0,0);
	\node at (YA) [left = 1mm of YA] {$Y_A$};
	\coordinate (YB) at (1,0,-1.7321);
	\node at (YB) [right = 1mm of YB] {$Y_B$};
	\coordinate (YC) at (1,0,1.7321);
	\node at (YC) [right = 1mm of YC] {$Y_C$};
	\draw[fill,opacity=0.3] (XA) 	-- (XB) --(XC) -- cycle;
	\draw[fill,opacity=0.3] (XA) 	-- (YB) --(YC) -- cycle;
	\draw[fill,opacity=0.3] (YA) 	-- (XB) --(YC) -- cycle;
	\draw[fill,opacity=0.3] (YA) 	-- (YB) --(XC) -- cycle;
	\coordinate (XA1) at (0.5,1,0);
	\node at (XA1) [right = 1mm of XA1] {$1$};
	\coordinate (XB1) at (-.25,1,0.43301);
	\node at (XB1) [left = 1mm of XB1] {$1$};
	\coordinate (XC1) at (-.25,1,-0.43301);
	\node at (XC1) [left = 1mm of XC1] {$1$};
	\coordinate (YA1) at (-2,1,0);
	\node at (YA1) [left = 1mm of YA1] {$1$};
	\coordinate (YB1) at (1,1,-1.7321);
	\node at (YB1) [right = 1mm of YB1] {$1$};
	\coordinate (YC1) at (1,1,1.7321);
	\node at (YC1) [right = 1mm of YC1] {$1$};
	\draw[dashed,gray] (XA1) 	-- (XB1) --(XC1) -- cycle;
	\draw[dashed,gray] (XA1) 	-- (YB1) --(YC1) -- cycle;
	\draw[dashed,gray] (YA1) 	-- (XB1) --(YC1) -- cycle;
	\draw[dashed,gray] (YA1) 	-- (YB1) --(XC1) -- cycle;
	\coordinate (XA0) at (0.5,2.6,0);
	\node at (XA0) [right = 1mm of XA0] {$0$};
	\coordinate (XB0) at (-.25,2.6,0.43301);
	\node at (XB0) [left = 1mm of XB0] {$0$};
	\coordinate (XC0) at (-.25,2.6,-0.43301);
	\node at (XC0) [left = 1mm of XC0] {$0$};
	\coordinate (YA0) at (-2,2.6,0);
	\node at (YA0) [left = 1mm of YA0] {$0$};
	\coordinate (YB0) at (1,2.6,-1.7321);
	\node at (YB0) [right = 1mm of YB0] {$0$};
	\coordinate (YC0) at (1,2.6,1.7321);
	\node at (YC0) [right = 1mm of YC0] {$0$};
	\draw[dashed,gray] (XA0) 	-- (XB0) --(XC0) -- cycle;
	\draw[dashed,gray] (XA0) 	-- (YB0) --(YC0) -- cycle;
	\draw[dashed,gray] (YA0) 	-- (XB0) --(YC0) -- cycle;
	\draw[dashed,gray] (YA0) 	-- (YB0) --(XC0) -- cycle;
	\draw[dashed] (XA) 	-- (XA0);
	\draw[dashed] (XB) 	-- (XB0);
	\draw[dashed] (XC) 	-- (XC0);
	\draw[dashed] (YA) 	-- (YA0);
	\draw[dashed] (YB) 	-- (YB0);
	\draw[dashed] (YC) 	-- (YC0);	
	\draw[blue,fill,opacity=0.3] (YA0) 	-- (YB0) --(XC0) -- cycle;
	\draw[blue,fill,opacity=0.3] (YA0) 	-- (YB1) --(XC1) -- cycle;
	\draw[blue,fill,opacity=0.3] (YA1) 	-- (YB0) --(XC1) -- cycle;
	\draw[blue,fill,opacity=0.3] (YA1) 	-- (YB1) --(XC0) -- cycle;
	\draw[violet,fill,opacity=0.3] (XA0) 	-- (XB0) --(XC1) -- cycle;
	\draw[violet,fill,opacity=0.3] (XA0) 	-- (XB1) --(XC0) -- cycle;
	\draw[violet,fill,opacity=0.3] (XA1) 	-- (XB0) --(XC0) -- cycle;
	\draw[violet,fill,opacity=0.3] (XA1) 	-- (XB1) --(XC1) -- cycle;
	\draw[teal,fill,opacity=0.3] (XA0) 	-- (YB0) --(YC1) -- cycle;
	\draw[teal,fill,opacity=0.3] (XA0) 	-- (YB1) --(YC0) -- cycle;
	\draw[teal,fill,opacity=0.3] (XA1) 	-- (YB0) --(YC0) -- cycle;
	\draw[teal,fill,opacity=0.3] (XA1) 	-- (YB1) --(YC1) -- cycle;
	\draw[orange,fill,opacity=0.3] (YA0) 	-- (XB0) --(YC1) -- cycle;
	\draw[orange,fill,opacity=0.3] (YA0) 	-- (XB1) --(YC0) -- cycle;
	\draw[orange,fill,opacity=0.3] (YA1) 	-- (XB0) --(YC0) -- cycle;
	\draw[orange,fill,opacity=0.3] (YA1) 	-- (XB1) --(YC1) -- cycle;
	\end{tikzpicture}\captionof{figure}{Bundle diagram with all contexts for the Anders-Browne example.}
\end{figure}
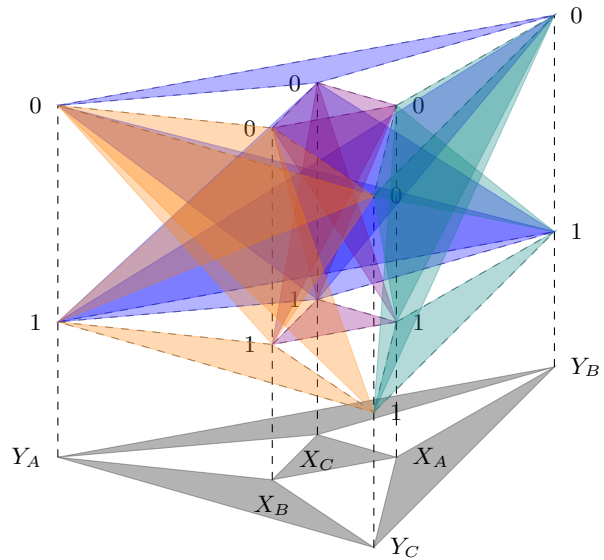
We obtain a generalised bundle diagram for the Anders-Browne example by depicting the possible outcomes of a joint measurement of a context not as a line, but instead as a triangular facet. We are lucky that in three dimensions we can always draw a surface through three points. To build, in a similar way, a bundle diagram for a contextual model involving four or more agents would require a higher dimensional ambient space. In the following picture the violet triangles represent the sections above the context $C_1$. The context $C_2$ is marked with teal triangles. The sections related to $C_3$ are depicted in orange, the ones belonging to $C_4$ in blue. Points in the stalk above a measurement are labelled according to the observable and the outcome, for example, the point above $X_A$ corresponding to outcome $1$ is denoted $X_A1$.

We now explore the contextuality of this model: if there exists a triangle which cannot be continued to a global section then the model is contextual. We argue this is the case by starting with the violet triangle $X_A1-X_B1-X_C1$ and exploring all potential compatible ways of extending this section to a global section. We illustrate this procedure via a tree diagram which graphically enumerates all the possible extensions. 
\onecolumngrid
\begin{center}
	\begin{tikzpicture}
		[level distance=1.3cm,
		level 1/.style={sibling distance=6cm},
		level 2/.style={sibling distance=3.5cm},
		level 3/.style={sibling distance=1.2cm},
		level 4/.style={sibling distance=.6cm},
		triangle/.style = {regular polygon, regular polygon sides=3 }]
		\tikzstyle{every node}=[font=\small]
		\node {
	\begin{tikzpicture}[level 1/.style={level distance=2.8cm,sibling distance=0.3cm,every child/.style={edge from parent/.style={draw,white}}},level 2/.style={level distance=1.3cm,sibling distance=0.3cm,every child/.style={edge from parent/.style={draw,black}}},level 3/.style={sibling distance=1cm},grow'=up]
		\node{} child{node{$X_A1$}child{node[triangle,rotate=180,fill=violet,opacity=0.3,text width=2mm]{}child {node [rotate=270]{$X_B1*$}}child{node[rotate=270] {$X_C1**$}}}};	
           \end{tikzpicture}}
		child {node[triangle,fill=violet,opacity=0.3] { }
			child {node {$X_B0\text{ }+$}}
			child {node[left=1cm] {$X_C0\text{ }+$}}
		}
		child {node[left=2cm,triangle,fill=teal,opacity=0.3] { }
			child {node {$Y_B0$}
				child {node[triangle,fill=teal,opacity=0.3] { }
					child {node[rotate=270] {$X_A0\text{ }+$}}
					child {node[rotate=270] {$Y_C1$}}
				}
				child {node[triangle,fill=blue,opacity=0.3] { }
					child {node[rotate=270] {$Y_A0$}}
					child {node[rotate=270] {$X_C0\text{ }+$}}
				}
				child {node[triangle,fill=blue,opacity=0.3] { }
					child {node[rotate=270] {$Y_A1$}
					child[orange,line width=1pt]}
					child {node[rotate=270] {$X_C1^{**}$}}
				}
			}
			child {node {$Y_C0$}
				child {node[triangle,fill=teal,opacity=0.3] { }
					child {node[rotate=270] {$X_A0\text{ }+$}}
					child {node[rotate=270] {$Y_B1$}}
				}
				child {node[triangle,fill=orange,opacity=0.3] { }
					child {node[rotate=270] {$Y_A0$}
					child[orange,line width=1pt]}
					child {node[rotate=270] {$X_B1^*$}}
				}
				child {node[triangle,fill=orange,opacity=0.3] { }
					child {node[rotate=270] {$Y_A1$}}
					child {node[rotate=270] {$X_B0\text{ }+$}}
				}
			}
		}
		child {node[left=1cm,triangle,fill=teal,opacity=0.3] { }
			child {node {$Y_B1$}
				child {node[triangle,fill=teal,opacity=0.3] { }
					child {node[rotate=270] {$X_A0\text{ }+$}}
					child {node[rotate=270] {$Y_C0$}}
				}
				child {node[triangle,fill=blue,opacity=0.3] { }
					child {node[rotate=270] {$Y_A0$}
						child[blue,line width=1pt]}
					child {node[rotate=270] {$X_C1^{**}$}}
				}
				child {node[triangle,fill=blue,opacity=0.3] { }
					child {node[rotate=270] {$Y_A1$}}
					child {node[rotate=270] {$X_C0\text{ }+$}}
				}
			}
			child {node {$Y_C1$}
			child {node[triangle,fill=teal,opacity=0.3] { }
				child {node[rotate=270] {$X_A0\text{ }+$}}
				child {node[rotate=270] {$Y_B0$}}
			}
			child {node[triangle,fill=orange,opacity=0.3] { }
				child {node[rotate=270] {$Y_A0$}}
				child {node[rotate=270] {$X_B0\text{ }+$}}
			}
			child {node[triangle,fill=orange,opacity=0.3] { }
				child {node[rotate=270] {$Y_A1$}
					child[blue,line width=1pt]}
				child {node[rotate=270] {$X_B1^*$}}
			}
			}
		};
		\end{tikzpicture}
	\begin{tikzpicture}\end{tikzpicture}
	
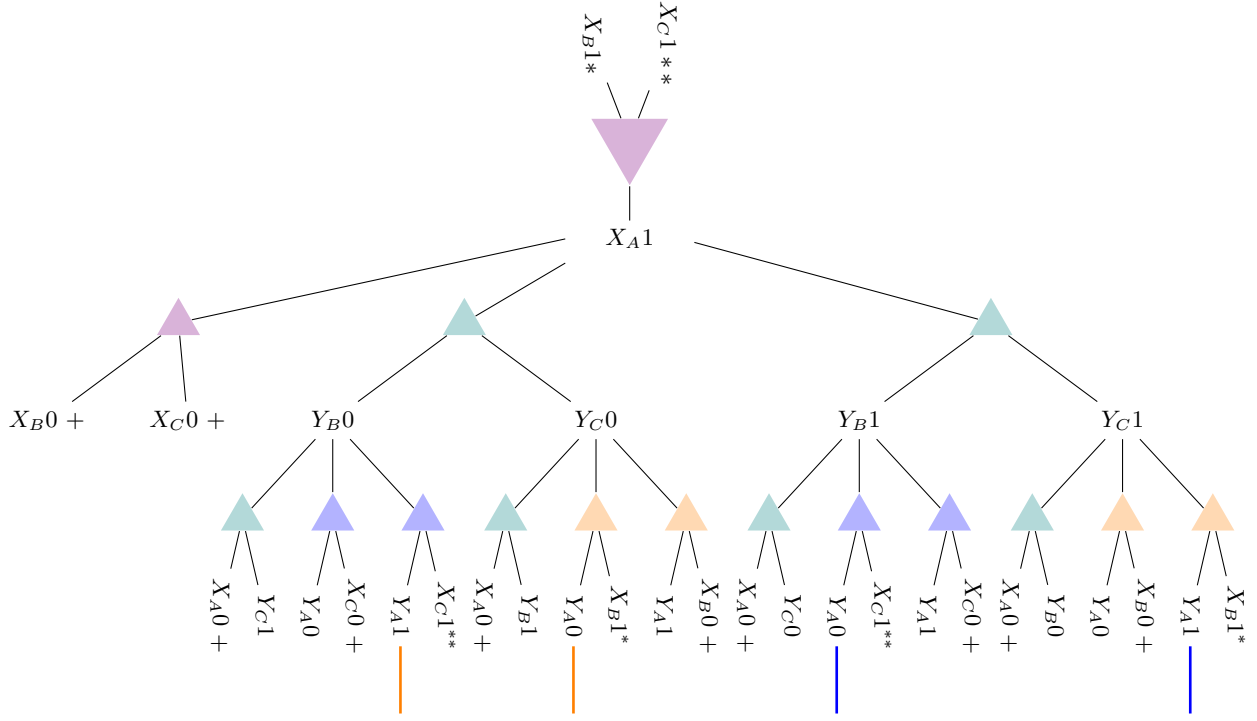
\captionof{figure}{Trying to build a global section with the triangle $X_A1-X_B1-X_C1$.} \label{firsttree}
\end{center}
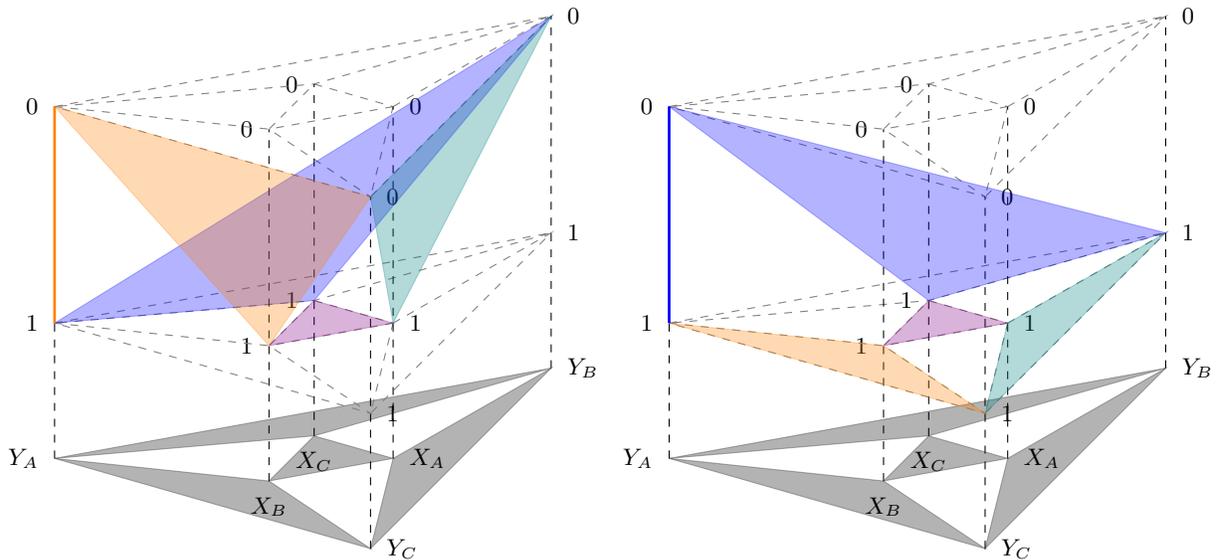
\begin{figure}[H]
	\centering
\begin{tikzpicture}[scale=1.8]
\coordinate (XA) at (0.5,0,0);
\node at (XA) [right = 1mm of XA] {$X_A$};
\coordinate (XB) at (-.25,0,0.43301);
\node at (XB) [below = 1mm of XB] {$X_B$};
\coordinate (XC) at (-.25,0,-0.43301);
\node at (XC) [below = 1mm of XC] {$X_C$};
\coordinate (YA) at (-2,0,0);
\node at (YA) [left = 1mm of YA] {$Y_A$};
\coordinate (YB) at (1,0,-1.7321);
\node at (YB) [right = 1mm of YB] {$Y_B$};
\coordinate (YC) at (1,0,1.7321);
\node at (YC) [right = 1mm of YC] {$Y_C$};
\draw[fill,opacity=0.3] (XA) 	-- (XB) --(XC) -- cycle;
\draw[fill,opacity=0.3] (XA) 	-- (YB) --(YC) -- cycle;
\draw[fill,opacity=0.3] (YA) 	-- (XB) --(YC) -- cycle;
\draw[fill,opacity=0.3] (YA) 	-- (YB) --(XC) -- cycle;
\coordinate (XA1) at (0.5,1,0);
\node at (XA1) [right = 1mm of XA1] {$1$};
\coordinate (XB1) at (-.25,1,0.43301);
\node at (XB1) [left = 1mm of XB1] {$1$};
\coordinate (XC1) at (-.25,1,-0.43301);
\node at (XC1) [left = 1mm of XC1] {$1$};
\coordinate (YA1) at (-2,1,0);
\node at (YA1) [left = 1mm of YA1] {$1$};
\coordinate (YB1) at (1,1,-1.7321);
\node at (YB1) [right = 1mm of YB1] {$1$};
\coordinate (YC1) at (1,1,1.7321);
\node at (YC1) [right = 1mm of YC1] {$1$};
\draw[dashed,gray] (XA1) 	-- (XB1) --(XC1) -- cycle;
\draw[dashed,gray] (XA1) 	-- (YB1) --(YC1) -- cycle;
\draw[dashed,gray] (YA1) 	-- (XB1) --(YC1) -- cycle;
\draw[dashed,gray] (YA1) 	-- (YB1) --(XC1) -- cycle;
\coordinate (XA0) at (0.5,2.6,0);
\node at (XA0) [right = 1mm of XA0] {$0$};
\coordinate (XB0) at (-.25,2.6,0.43301);
\node at (XB0) [left = 1mm of XB0] {$0$};
\coordinate (XC0) at (-.25,2.6,-0.43301);
\node at (XC0) [left = 1mm of XC0] {$0$};
\coordinate (YA0) at (-2,2.6,0);
\node at (YA0) [left = 1mm of YA0] {$0$};
\coordinate (YB0) at (1,2.6,-1.7321);
\node at (YB0) [right = 1mm of YB0] {$0$};
\coordinate (YC0) at (1,2.6,1.7321);
\node at (YC0) [right = 1mm of YC0] {$0$};
\draw[dashed,gray] (XA0) 	-- (XB0) --(XC0) -- cycle;
\draw[dashed,gray] (XA0) 	-- (YB0) --(YC0) -- cycle;
\draw[dashed,gray] (YA0) 	-- (XB0) --(YC0) -- cycle;
\draw[dashed,gray] (YA0) 	-- (YB0) --(XC0) -- cycle;
\draw[dashed] (XA) 	-- (XA0);
\draw[dashed] (XB) 	-- (XB0);
\draw[dashed] (XC) 	-- (XC0);
\draw[dashed] (YA) 	-- (YA0);
\draw[dashed] (YB) 	-- (YB0);
\draw[dashed] (YC) 	-- (YC0);
\draw[blue,fill,opacity=0.3] (YA1) 	-- (YB0) --(XC1) -- cycle;
\draw[violet,fill,opacity=0.3] (XA1) 	-- (XB1) --(XC1) -- cycle;
\draw[teal,fill,opacity=0.3] (XA1) 	-- (YB0) --(YC0) -- cycle;
\draw[orange,fill,opacity=0.3] (YA0) 	-- (XB1) --(YC0) -- cycle;
\draw[orange,line width=1pt] (YA1) 	-- (YA0);
\end{tikzpicture}
	\begin{tikzpicture}[scale=1.8]
\coordinate (XA) at (0.5,0,0);
\node at (XA) [right = 1mm of XA] {$X_A$};
\coordinate (XB) at (-.25,0,0.43301);
\node at (XB) [below = 1mm of XB] {$X_B$};
\coordinate (XC) at (-.25,0,-0.43301);
\node at (XC) [below = 1mm of XC] {$X_C$};
\coordinate (YA) at (-2,0,0);
\node at (YA) [left = 1mm of YA] {$Y_A$};
\coordinate (YB) at (1,0,-1.7321);
\node at (YB) [right = 1mm of YB] {$Y_B$};
\coordinate (YC) at (1,0,1.7321);
\node at (YC) [right = 1mm of YC] {$Y_C$};
\draw[fill,opacity=0.3] (XA) 	-- (XB) --(XC) -- cycle;
\draw[fill,opacity=0.3] (XA) 	-- (YB) --(YC) -- cycle;
\draw[fill,opacity=0.3] (YA) 	-- (XB) --(YC) -- cycle;
\draw[fill,opacity=0.3] (YA) 	-- (YB) --(XC) -- cycle;
\coordinate (XA1) at (0.5,1,0);
\node at (XA1) [right = 1mm of XA1] {$1$};
\coordinate (XB1) at (-.25,1,0.43301);
\node at (XB1) [left = 1mm of XB1] {$1$};
\coordinate (XC1) at (-.25,1,-0.43301);
\node at (XC1) [left = 1mm of XC1] {$1$};
\coordinate (YA1) at (-2,1,0);
\node at (YA1) [left = 1mm of YA1] {$1$};
\coordinate (YB1) at (1,1,-1.7321);
\node at (YB1) [right = 1mm of YB1] {$1$};
\coordinate (YC1) at (1,1,1.7321);
\node at (YC1) [right = 1mm of YC1] {$1$};
\draw[dashed,gray] (XA1) 	-- (XB1) --(XC1) -- cycle;
\draw[dashed,gray] (XA1) 	-- (YB1) --(YC1) -- cycle;
\draw[dashed,gray] (YA1) 	-- (XB1) --(YC1) -- cycle;
\draw[dashed,gray] (YA1) 	-- (YB1) --(XC1) -- cycle;
\coordinate (XA0) at (0.5,2.6,0);
\node at (XA0) [right = 1mm of XA0] {$0$};
\coordinate (XB0) at (-.25,2.6,0.43301);
\node at (XB0) [left = 1mm of XB0] {$0$};
\coordinate (XC0) at (-.25,2.6,-0.43301);
\node at (XC0) [left = 1mm of XC0] {$0$};
\coordinate (YA0) at (-2,2.6,0);
\node at (YA0) [left = 1mm of YA0] {$0$};
\coordinate (YB0) at (1,2.6,-1.7321);
\node at (YB0) [right = 1mm of YB0] {$0$};
\coordinate (YC0) at (1,2.6,1.7321);
\node at (YC0) [right = 1mm of YC0] {$0$};
\draw[dashed,gray] (XA0) 	-- (XB0) --(XC0) -- cycle;
\draw[dashed,gray] (XA0) 	-- (YB0) --(YC0) -- cycle;
\draw[dashed,gray] (YA0) 	-- (XB0) --(YC0) -- cycle;
\draw[dashed,gray] (YA0) 	-- (YB0) --(XC0) -- cycle;
\draw[dashed] (XA) 	-- (XA0);
\draw[dashed] (XB) 	-- (XB0);
\draw[dashed] (XC) 	-- (XC0);
\draw[dashed] (YA) 	-- (YA0);
\draw[dashed] (YB) 	-- (YB0);
\draw[dashed] (YC) 	-- (YC0);
\draw[violet,fill,opacity=0.3] (XA1) 	-- (XB1) --(XC1) -- cycle;
\draw[teal,fill,opacity=0.3] (XA1) 	-- (YB1) --(YC1) -- cycle;
\draw[orange,fill,opacity=0.3] (YA1) 	-- (XB1) --(YC1) -- cycle;
\draw[blue,fill,opacity=0.3] (YA0) 	-- (YB1) --(XC1) -- cycle;
\draw[blue,line width=1pt] (YA1) 	-- (YA0);
\end{tikzpicture}
	\captionof{figure}{Examples for failed global sections for the Anders-Browne example.}\label{twobundles}
\end{figure}
\twocolumngrid

We draw a node (connected to our start triangle) in the tree diagram for the assignment of $1$ to $X_A$, i.e., the point $X_A1$ in the stalk above $X_A$. We then draw leaves connected to nodes labelled by triangles containing $X_A1$. Each such triangle node has leaves corresponding to compatible assignments and so on. Thus our tree is comprised of alternating layers of triangles and assignments. We tag the leaves of a node with a $+$ symbol when they are incompatible. Since in a global section only one outcome per measurement is allowed and our start triangle contains $X_A1$, $X_B1$, and $X_C1$ we stop when we meet any of $X_A0$, $X_B0$, or $X_C0$. 

Of course we also have to connect the edges $X_B1$ and $X_C1$ of the starting triangle correctly. When we meet them again via this extension process we mark the nodes with $*$ or $**$. We can see that there is a possibility to connect the edges $X_B1$ and $X_C1$ of the start triangle and touch all measurements $X_A,X_B,Y_A$ and $Y_B$ with ones, unless $Y_A1\neq Y_A0$. This situation is tagged with an orange line. A similar case corresponds to the teal lines. We can see graphically how building a global section fails in these two cases in Figure \ref{twobundles}.

Using such a tree diagram we can also see that it is not possible to build a complete global section with the triangle $X_A1-X_B0-X_C0$.

According to this argument the model is logically contextual. We can further deduce that the model is strongly contextual, i.e.\ there is no global section at all. To this end we define a function $f: \text{Support}[P(O)] \rightarrow \{\pm1\}$ for each context, where $P(O)$ is the possibility to get the outcome $O=(a,b,c)$ via the joint measurement of the context. Considering $C_1$, the function is defined by
\begin{align*}
f_{C_1}(X_A=a,X_B=b,X_C=c)&=(-1)^a(-1)^b(-1)^c\\
\intertext{and produces the following codomain}
f_{C_1}(0,0,1)&=-1\\
f_{C_1}(0,1,0)&=-1\\
f_{C_1}(1,0,0)&=-1\\
f_{C_1}(1,1,1)&=-1\text{.}
\end{align*}
Defining the other functions analogously and computing the codomains we observe that the codomain for one particular context is either $\{-1\}$ or $\{1\}$. So we can summarize the functions $f_{C_i}$ for $i\in \{1,\dots,4\}$ with an overall parity function $F: C_i \rightarrow \{\pm1\}$ defined on $\{C_i|i\in \{1,\dots,4\}\}$. We depict the values of the parity function for each context in Figure \ref{bottomparity}.
\begin{figure}[H]
	\centering
	\begin{tikzpicture}[scale=1.4]
	\coordinate (XA) at (0.5,0);
	\node at (XA) [above = 1mm of XA] {$X_A$};
	\coordinate (XB) at (-.25,0.43301);
	\node at (XB) [right = 1mm of XB] {$X_B$};
	\coordinate (XC) at (-.25,-0.43301);
	\node at (XC) [right = 1mm of XC] {$X_C$};
	\coordinate (YA) at (-2,0);
	\node at (YA) [left = 1mm of YA] {$Y_A$};
	\coordinate (YB) at (1,-1.7321);
	\node at (YB) [left = 1mm of YB] {$Y_B$};
	\coordinate (YC) at (1,1.7321);
	\node at (YC) [right = 1mm of YC] {$Y_C$};
	\draw[fill,opacity=0.3, violet] (XA) 	-- (XB) --(XC) -- cycle;
	\draw[fill,opacity=0.3, teal] (XA) 	-- (YB) --(YC) -- cycle;
	\draw[fill,opacity=0.3, orange] (YA) 	-- (XB) --(YC) -- cycle;
	\draw[fill,opacity=0.3, blue] (YA) 	-- (YB) --(XC) -- cycle;
	\draw(0:0) node{-1};
	\draw(120:.75) node{-1};
	\draw(240:.75) node{1};
	\draw(0:.75) node{-1};
	\end{tikzpicture}
	\captionof{figure}{Values of parity function.} \label{bottomparity}
\end{figure}
It follows that all triangles have an odd number of edges on the $1$-plane, except for the blue ones. So it is not possible to built any global section. Also the structure of the NAND gate becomes clear in this approach, i.e. 
\begin{align*}
X_AX_BX_C \ket{\Psi}&=-\ket{\Psi}\\
X_AY_BY_C \ket{\Psi}&=-\ket{\Psi}\\
Y_AX_BY_C \ket{\Psi}&=-\ket{\Psi}\\
Y_AY_BX_C \ket{\Psi}&=\ket{\Psi}\text{.}
\end{align*}
We were not able to find a global section using the state $\ket{\Psi}=\frac{1}{\sqrt{2}}(\ket{001}-\ket{110})$. As a next step we want to figure out whether there exists a state which produces a model with a global section using the same contexts. Thus we search for a model which fulfills $F(C_i)=-1$ for an even number of contexts $C_i$. An appropriate set of conditions on $\ket{\Psi}$ is depicted in Figure \ref{bottomparitysearch} and requires
\begin{align*}
X_AX_BX_C \ket{\Psi}&=-\ket{\Psi}\\
X_AY_BY_C \ket{\Psi}&=\ket{\Psi}\\
Y_AX_BY_C \ket{\Psi}&=-\ket{\Psi}\\
Y_AY_BX_C \ket{\Psi}&=\ket{\Psi}\text{.}
\end{align*}
\begin{figure}[H]
	\centering
	\begin{tikzpicture}[scale=1.4]
	\coordinate (XA) at (0.5,0);
	\node at (XA) [above = 1mm of XA] {$X_A$};
	\coordinate (XB) at (-.25,0.43301);
	\node at (XB) [right = 1mm of XB] {$X_B$};
	\coordinate (XC) at (-.25,-0.43301);
	\node at (XC) [right = 1mm of XC] {$X_C$};
	\coordinate (YA) at (-2,0);
	\node at (YA) [left = 1mm of YA] {$Y_A$};
	\coordinate (YB) at (1,-1.7321);
	\node at (YB) [left = 1mm of YB] {$Y_B$};
	\coordinate (YC) at (1,1.7321);
	\node at (YC) [right = 1mm of YC] {$Y_C$};
	\draw[fill,opacity=0.3, violet] (XA) 	-- (XB) --(XC) -- cycle;
	\draw[fill,opacity=0.3, teal] (XA) 	-- (YB) --(YC) -- cycle;
	\draw[fill,opacity=0.3, orange] (YA) 	-- (XB) --(YC) -- cycle;
	\draw[fill,opacity=0.3, blue] (YA) 	-- (YB) --(XC) -- cycle;
	\draw(0:0) node{-1};
	\draw(120:.75) node{-1};
	\draw(240:.75) node{1};
	\draw(0:.75) node{1};
	\end{tikzpicture}
	\captionof{figure}{Values of parity function.} \label{bottomparitysearch}
\end{figure}
To build such a state $\ket{\Psi}$ fulfilling these conditions we construct the projectors
\begin{align*}
P_1&=\frac{\mathbbm{1}-X_AX_BX_C}{2}\\
P_2&=\frac{\mathbbm{1}+X_AY_BY_C}{2}\\
P_3&=\frac{\mathbbm{1}-Y_AY_BX_C}{2}\\
P_4&=\frac{\mathbbm{1}+X_AX_BY_C}{2}\text{.}
\end{align*}
If $P_1\ket{\Psi}=P_2\ket{\Psi}=P_3\ket{\Psi}=P_4\ket{\Psi}=\ket{\Psi}$ then $\text{tr}(P_1P_2P_3P_4)=1$ and therefore $P_1P_2P_3P_4 =\ket{\Psi}\bra{\Psi}$ \cite{Nielsen2000}. In the case depicted in Figure \ref{bottomparitysearch} we get $P_1P_2P_3P_4=0$. Hence there is no such state $\ket{\Psi}\ne0$. 

Calculating the states resulting from all sets of conditions satisfying $F(C_i)=-1$ for an even number of the four contexts $C_i$ analogously, we notice an interesting 
\begin{lem}
The states with an even number of minus signs in the projectors vanish using the contexts \begin{align*}C_1&=\{X_A,X_B,X_C\}\\C_2&=\{X_A,Y_B,Y_C\}\\C_3&=\{Y_A,X_B,Y_C\}\\C_4&=\{Y_A,Y_B,X_C\}\text{.}\end{align*} It is only possible to find a state $\ket{\Psi}\ne0$ for an odd number of minus signs in the projectors.\label{oddnumber}
\end{lem}

\begin{proof}
	We can prove this by calculating $P_1 P_2 P
	_3 P_4 = \ket{\Psi}\bra{\Psi}$ for all cases.
\end{proof}
On one hand we reasoned that for a state with an odd number of minus signs there is no global face. On the other we proved that there is no state $\ket{\Psi}\ne0$ satisfying $F(C_i)=-1$ for an even number of $C_i$. Thus there is no state producing an empirical model with a global section using the contexts defined at the beginning of the section.

%% file: 4_ClusterRing.tex
\section{The cluster state on a ring of $5$ qubits.}
We consider now a cluster state model on a ring of $n=5$ qubits. This state is stabilized by five stabiliser operators (the first five below). Taking products we generate five additional stabilisers: 
\begin{align*}
	XZ\mathbbm{1}\mathbbm{1}Z\ket{\Psi}&=\ket{\Psi}&
ZXZ\mathbbm{1}\mathbbm{1}\ket{\Psi}&=\ket{\Psi}\\
\mathbbm{1}ZXZ\mathbbm{1}\ket{\Psi}&=\ket{\Psi}&
\mathbbm{1}\mathbbm{1}ZXZ\ket{\Psi}&=\ket{\Psi}\\
Z\mathbbm{1}\mathbbm{1}ZX\ket{\Psi}&=\ket{\Psi}&
ZX\mathbbm{1}XZ\ket{\Psi}&=\ket{\Psi}\\
ZZX\mathbbm{1}X\ket{\Psi}&=\ket{\Psi}&
XZZX\mathbbm{1}\ket{\Psi}&=\ket{\Psi}\\
\mathbbm{1}XZZX\ket{\Psi}&=\ket{\Psi}&
X\mathbbm{1}XZZ\ket{\Psi}&=\ket{\Psi}
\end{align*}
We use for our contexts the observables contained in each of these $10$ stabilizers. These are listed in Table \ref{cluster1D5}. 
The entries with nonzero probabilities are used for building the bundle diagram. For a better overview we depict the base of the bundle diagram and the bundle diagram itself in two different figures, divided in the contexts $C_1$ to $C_5$ and $C_6$ to $C_{10}$. (The complete bundle diagram is the union of these two figures.) Note that since we have contexts involving $4$ observables we are unable to depict the sections above the corresponding quadrilaterals as flat surfaces; they are buckled in general.
\begin{figure}[H]
	\centering
\resizebox{\columnwidth}{!}{
	\begin{tabular}{c |c c c c c | c c c c c c c c }
		&$1$ & $2$ & $3$ & $4$ &$5$ &$ 0 0 0 $ & $ 0 0 1 $ & $ 0 1 0 $ & $ 0 1 1 $ & $ 1 0 0 $ & $ 1 0 1 $ & $ 1 1 0 $ & $ 1 1 1 $\\
		\hline
		$C_1$ &$X_1$ & $Z_2$  &$\mathbbm{1}$& $\mathbbm{1}$&    $Z_5$ &$1/4$ &  $0$  &  $0$ & $1/4$ &     $0$  &   $1/4$ &   $1/4$ & $0$ \\
		$C_2$ &$Z_1$ & $X_2$ & $Z_3$&  $\mathbbm{1}$& $\mathbbm{1}$& $1/4$ &  $0$  &  $0$ & $1/4$ &     $0$  &   $1/4$ &   $1/4$ & $0$ \\
		$C_3$ &$\mathbbm{1}$& $Z_2$ & $X_3$  &   $Z_4$ &  $\mathbbm{1}$&  $1/4$ &  $0$  &  $0$ & $1/4$ &     $0$  &   $1/4$ &   $1/4$ & $0$ \\
		$C_4$  &$\mathbbm{1}$& $\mathbbm{1}$& $Z_3$ & $X_4$  &   $Z_5$ &   $1/4$ &  $0$  &  $0$ & $1/4$ &     $0$  &   $1/4$ &   $1/4$ & $0$ \\
		$C_5$ &$Z_1$ & $\mathbbm{1}$& $\mathbbm{1}$&$Z_4$  &   $X_5$&   $1/4$ &  $0$  &  $0$ & $1/4$ &     $0$  &   $1/4$ &   $1/4$ & $0$ 
\end{tabular}}

\vspace{7mm}

	\resizebox{\columnwidth}{!}{
\begin{tabular}{c |c c c c c | c c c c c c c c }
		&$1$ & $2$ & $3$ & $4$ &$5$ &$0 0 0 0 $ & $0 0 0 1 $ & $0 0 1 0 $ & $0 0 1 1 $ & $0 1 0 0 $ & $0 1 0 1 $ & $0 1 1 0 $ & $0 1 1 1 $\\
		\hline
		$C_6$ &$Z_1$ & $X_2$  &$\mathbbm{1}$& $X_4$&    $Z_5$ &$1/4$ &  $0$  &  $0$ & $1/4$ &     $0$  &   $1/4$ &   $1/4$ & $0$ \\
		$C_7$&$Z_1$ & $Z_2$  &$X_3$&$\mathbbm{1}$&    $X_5$& $1/4$ &  $0$  &  $0$ & $1/4$ &     $0$  &   $1/4$ &   $1/4$ & $0$ \\
		$C_8$ &$X_1$& $Z_2$ & $Z_3$  &   $X_4$ &  $\mathbbm{1}$&  $1/4$ &  $0$  &  $0$ & $1/4$ &     $0$  &   $1/4$ &   $1/4$ & $0$ \\
		$C_9$  &$\mathbbm{1}$& $X_2$& $Z_3$ & $Z_4$  &   $X_5$ &   $1/4$ &  $0$  &  $0$ & $1/4$ &     $0$  &   $1/4$ &   $1/4$ & $0$ \\
		$C_{10}$ &$X_1$ & $\mathbbm{1}$& $X_3$&$Z_4$  &   $Z_5$&   $1/4$ &  $0$  &  $0$ & $1/4$ &     $0$  &   $1/4$ &   $1/4$ & $0$\\ 
			 & &  & & & &$1 0 0 0 $ & $1 0 0 1 $ & $1 0 1 0 $ & $1 0 1 1 $ & $1 1 0 0 $ & $1 1 0 1 $ & $1 1 1 0 $ & $1 1 1 1 $\\
		\hline
		$C_6$ &$Z_1$ & $X_2$  &$\mathbbm{1}$&$X_4$&    $Z_5$ &$0$ &  $1/4$  &  $1/4$ & $0$ &     $1/4$  &   $0$ &   $0$ & $1/4$ \\
		$C_7$&$Z_1$ & $Z_2$  &$X_3$&$\mathbbm{1}$&    $X_5$ &$0$ &  $1/4$  &  $1/4$ & $0$ &     $1/4$  &   $0$ &   $0$ & $1/4$ \\
		$C_8$ &$X_1$& $Z_2$ & $Z_3$  &   $X_4$ &  $\mathbbm{1}$&  $0$ &  $1/4$  &  $1/4$ & $0$ &     $1/4$  &   $0$ &   $0$ & $1/4$ \\
		$C_9$  &$\mathbbm{1}$& $X_2$& $Z_3$ & $Z_4$  &   $X_5$ &   $0$ &  $1/4$  &  $1/4$ & $0$ &     $1/4$  &   $0$ &   $0$ & $1/4$ \\
		$C_{10}$ &$X_1$ & $\mathbbm{1}$& $X_3$&$Z_4$  &   $Z_5$&  $0$ &  $1/4$  &  $1/4$ & $0$ &     $1/4$  &   $0$ &   $0$ & $1/4$ 
\end{tabular}}
\captionof{table}{Empirical model for cluster state $n=5$.}\label{cluster1D5}
\end{figure}
\begin{figure}[H]
	\centering
	\begin{tikzpicture}[scale=0.8]
	\coordinate (X1) at (.7,0);
	\node at (X1) [right = .01mm of X1] {$X_1$};
	\coordinate (X2) at (0.21631,0.66574);
	\node at (X2) [above = .01mm of X2] {$X_2$};
	\coordinate (X3) at (-0.56631,0.41145);
	\node at (X3) [above left = .01mm of X3] {$X_3$};
	\coordinate (X4) at (-0.56631,-0.41145);
	\node at (X4) [below left = .01mm of X4] {$X_4$};
	\coordinate (X5) at (0.21631,-0.66574);
	\node at (X5) [below = .01mm of X5] {$X_5$};
	\coordinate (Z1) at (3.5,0);
	\node at (Z1) [right = 1mm of Z1] {$Z_1$};
	\coordinate (Z2) at (1.0816,3.3287);
	\node at (Z2) [above = 1mm of Z2] {$Z_2$};
	\coordinate (Z3) at (-2.8316,2.0572);
	\node at (Z3) [above = 1mm of Z3] {$Z_3$};
	\coordinate (Z4) at (-2.8316,-2.0572);
	\node at (Z4) [left = 1mm of Z4] {$Z_4$};
	\coordinate (Z5) at (1.0816,-3.3287);
	\node at (Z5) [below = 1mm of Z5] {$Z_5$};
	\draw[fill,opacity=0.3,teal] (Z1) -- (X2) -- (Z3) -- cycle;
	\draw[fill,opacity=0.3,orange] (Z2) -- (X3) -- (Z4) -- cycle;
	\draw[fill,opacity=0.3,blue] (Z3) -- (X4) -- (Z5) -- cycle;
	\draw[fill,opacity=0.3,pink] (Z4) -- (X5) -- (Z1) -- cycle;
	\draw[fill,opacity=0.3,violet] (Z5) -- (X1) -- (Z2) -- cycle;
\end{tikzpicture}\captionof{figure}{Base of bundle diagram for the cluster state on a ring of $n=5$ qubits involving contexts $C_1-C_5$.}
\end{figure}
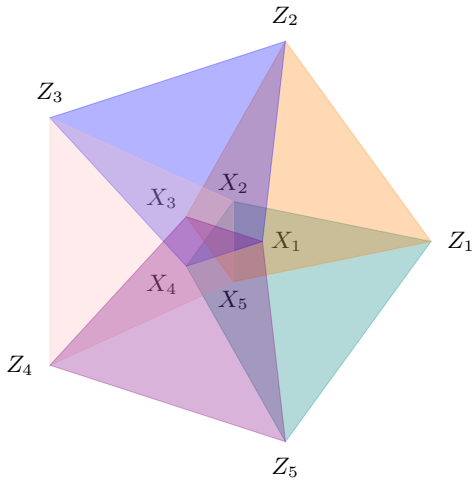
\begin{figure}[H]
	\centering
	\begin{tikzpicture}[scale=0.8]
	\coordinate (X1) at (.7,0);
	\node at (X1) [right = .01mm of X1] {$X_1$};
	\coordinate (X2) at (0.21631,0.66574);
	\node at (X2) [above = .01mm of X2] {$X_2$};
	\coordinate (X3) at (-0.56631,0.41145);
	\node at (X3) [above left = .01mm of X3] {$X_3$};
	\coordinate (X4) at (-0.56631,-0.41145);
	\node at (X4) [below left = .01mm of X4] {$X_4$};
	\coordinate (X5) at (0.21631,-0.66574);
	\node at (X5) [below = .01mm of X5] {$X_5$};
	\coordinate (Z1) at (3.5,0);
	\node at (Z1) [right = 1mm of Z1] {$Z_1$};
	\coordinate (Z2) at (1.0816,3.3287);
	\node at (Z2) [above = 1mm of Z2] {$Z_2$};
	\coordinate (Z3) at (-2.8316,2.0572);
	\node at (Z3) [above = 1mm of Z3] {$Z_3$};
	\coordinate (Z4) at (-2.8316,-2.0572);
	\node at (Z4) [left = 1mm of Z4] {$Z_4$};
	\coordinate (Z5) at (1.0816,-3.3287);
	\node at (Z5) [below = 1mm of Z5] {$Z_5$};
\draw[fill,opacity=0.3,teal] (Z1) -- (X2) -- (X4)-- (Z5) -- cycle;
\draw[fill,opacity=0.3,orange] (Z1)--(Z2)--(X3)--(X5) -- cycle;
\draw[fill,opacity=0.3,blue] (X1)--(Z2)--(Z3)--(X4) -- cycle;
\draw[fill,opacity=0.3,pink] (X2)--(Z3)--(Z4)--(X5) -- cycle;
\draw[fill,opacity=0.3,violet] (X1)--(X3)--(Z4)--(Z5)-- cycle;
	\end{tikzpicture}\captionof{figure}{Base of bundle diagram for the cluster state on a ring of $n=5$ qubits involving contexts $C_6-C_{10}$.}
\end{figure}
If there exists a triangle or square which cannot be continued to a global section then the model is contextual. We prove this by considering the teal triangle $Z_10-X_20-Z_30$ and depicting all ways of trying to find a global section via a tree diagram, as before. We start with the chosen triangle and draw all triangles containing $X_20$. We can see that only two of the first five children, $A$ and $B$, do not contain a $+$ symbol. All vertices of $A$ should be feasible for building a global section with these children. But continuing $X_50$ one can see that at the end there are no compatible triangles or squares. The case $B$ is similar. We hence cannot build a global section for the triangle $Z_10-X_20-Z_30$ and therefore the model is contextual.

\begin{figure}[H]
	\centering
	\begin{tikzpicture}[scale=1.1]
	\coordinate (X1) at (.7,0,0);
	\node at (X1) [right = .01mm of X1] {$X_1$};
	\coordinate (X2) at (0.21631,0,0.66574);
	\node at (X2) [above = .01mm of X2] {$X_2$};
	\coordinate (X3) at (-0.56631,0,0.41145);
	\node at (X3) [above left = .01mm of X3] {$X_3$};
	\coordinate (X4) at (-0.56631,0,-0.41145);
	\node at (X4) [below left = .01mm of X4] {$X_4$};
	\coordinate (X5) at (0.21631,0,-0.66574);
	\node at (X5) [below = .01mm of X5] {$X_5$};
	\coordinate (Z1) at (3.5,0,0);
	\node at (Z1) [right = 1mm of Z1] {$Z_1$};
	\coordinate (Z2) at (1.0816,0,3.3287);
	\node at (Z2) [above = 1mm of Z2] {$Z_2$};
	\coordinate (Z3) at (-2.8316,0,2.0572);
	\node at (Z3) [above = 1mm of Z3] {$Z_3$};
	\coordinate (Z4) at (-2.8316,0,-2.0572);
	\node at (Z4) [left = 1mm of Z4] {$Z_4$};
	\coordinate (Z5) at (1.0816,0,-3.3287);
	\node at (Z5) [below = 1mm of Z5] {$Z_5$};
	\draw[fill,opacity=0.3] (Z1) -- (X2) -- (Z3) -- cycle;
	\draw[fill,opacity=0.3] (Z2) -- (X3) -- (Z4) -- cycle;
	\draw[fill,opacity=0.3] (Z3) -- (X4) -- (Z5) -- cycle;
	\draw[fill,opacity=0.3] (Z4) -- (X5) -- (Z1) -- cycle;
	\draw[fill,opacity=0.3] (Z5) -- (X1) -- (Z2) -- cycle;
	\coordinate (X11) at (.7,1,0);
	\node at (X11) [right = .01mm of X11] {$1$};
	\coordinate (X21) at (0.21631,1,0.66574);
	\node at (X21) [above = .01mm of X21] {$1$};
	\coordinate (X31) at (-0.56631,1,0.41145);
	\node at (X31) [above left = .01mm of X31] {$1$};
	\coordinate (X41) at (-0.56631,1,-0.41145);
	\node at (X41) [below left = .01mm of X41] {$1$};
	\coordinate (X51) at (0.21631,1,-0.66574);
	\node at (X51) [below = .01mm of X51] {$1$};
	\coordinate (Z11) at (3.5,1,0);
	\node at (Z11) [right = 1mm of Z11] {$1$};
	\coordinate (Z21) at (1.0816,1,3.3287);
	\node at (Z21) [above = 1mm of Z21] {$1$};
	\coordinate (Z31) at (-2.8316,1,2.0572);
	\node at (Z31) [above = 1mm of Z31] {$1$};
	\coordinate (Z41) at (-2.8316,1,-2.0572);
	\node at (Z41) [left = 1mm of Z41] {$1$};
	\coordinate (Z51) at (1.0816,1,-3.3287);
	\node at (Z51) [below = 1mm of Z51] {$1$};
	\draw[dashed, gray] (Z11) -- (X21) -- (Z31) -- cycle;
	\draw[dashed, gray] (Z21) -- (X31) -- (Z41) -- cycle;
	\draw[dashed, gray] (Z31) -- (X41) -- (Z51) -- cycle;
	\draw[dashed, gray] (Z41) -- (X51) -- (Z11) -- cycle;
	\draw[dashed, gray] (Z51) -- (X11) -- (Z21) -- cycle;
	\coordinate (X10) at (.7,2.5,0);
	\node at (X10) [right = .01mm of X10] {$0$};
	\coordinate (X20) at (0.21631,2.5,0.66574);
	\node at (X20) [above = .01mm of X20] {$0$};
	\coordinate (X30) at (-0.56631,2.5,0.41145);
	\node at (X30) [above left = .01mm of X30] {$0$};
	\coordinate (X40) at (-0.56631,2.5,-0.41145);
	\node at (X40) [below left = .01mm of X40] {$0$};
	\coordinate (X50) at (0.21631,2.5,-0.66574);
	\node at (X50) [below = .01mm of X50] {$0$};
	\coordinate (Z10) at (3.5,2.5,0);
	\node at (Z10) [right = 1mm of Z10] {$0$};
	\coordinate (Z20) at (1.0816,2.5,3.3287);
	\node at (Z20) [above = 1mm of Z20] {$0$};
	\coordinate (Z30) at (-2.8316,2.5,2.0572);
	\node at (Z30) [above = 1mm of Z30] {$0$};
	\coordinate (Z40) at (-2.8316,2.5,-2.0572);
	\node at (Z40) [left = 1mm of Z40] {$0$};
	\coordinate (Z50) at (1.0816,2.5,-3.3287);
	\node at (Z50) [below = 1mm of Z50] {$2$};

	\draw[dashed, gray] (Z10) -- (X20) -- (Z30) -- cycle;
	\draw[dashed, gray] (Z20) -- (X30) -- (Z40) -- cycle;
	\draw[dashed, gray] (Z30) -- (X40) -- (Z50) -- cycle;
	\draw[dashed, gray] (Z40) -- (X50) -- (Z10) -- cycle;
	\draw[dashed, gray] (Z50) -- (X10) -- (Z20) -- cycle;
	\draw[dashed, gray] (X1) -- (X10);
	\draw[dashed, gray] (X2) -- (X20);
	\draw[dashed, gray] (X3) -- (X30);	
	\draw[dashed, gray] (X4) -- (X40);	
	\draw[dashed, gray] (X5) -- (X50);
	\draw[dashed, gray] (Z1) -- (Z10);	
	\draw[dashed, gray] (Z2) -- (Z20);	
	\draw[dashed, gray] (Z3) -- (Z30);	
	\draw[dashed, gray] (Z4) -- (Z40);
	\draw[dashed, gray] (Z5) -- (Z50);
	\draw[fill,opacity=0.3,teal] (Z10) -- (X20) -- (Z30) -- cycle;
	\draw[fill,opacity=0.3,teal] (Z10) -- (X21) -- (Z31) -- cycle;
	\draw[fill,opacity=0.3,teal] (Z11) -- (X20) -- (Z31) -- cycle;
	\draw[fill,opacity=0.3,teal] (Z11) -- (X21) -- (Z30) -- cycle;
	
	\draw[fill,opacity=0.3,orange] (Z20) -- (X30) -- (Z40) -- cycle;
	\draw[fill,opacity=0.3,orange] (Z20) -- (X31) -- (Z41) -- cycle;
	\draw[fill,opacity=0.3,orange] (Z21) -- (X30) -- (Z41) -- cycle;
	\draw[fill,opacity=0.3,orange] (Z21) -- (X31) -- (Z40) -- cycle;
	
	\draw[fill,opacity=0.3,blue] (Z30) -- (X40) -- (Z50) -- cycle;
	\draw[fill,opacity=0.3,blue] (Z30) -- (X41) -- (Z51) -- cycle;
	\draw[fill,opacity=0.3,blue] (Z31) -- (X40) -- (Z51) -- cycle;
	\draw[fill,opacity=0.3,blue] (Z31) -- (X41) -- (Z50) -- cycle;
	
	\draw[fill,opacity=0.3,pink] (Z10) -- (Z40) -- (X50) -- cycle;
	\draw[fill,opacity=0.3,pink] (Z10) -- (Z41) -- (X51) -- cycle;
	\draw[fill,opacity=0.3,pink] (Z11) -- (Z40) -- (X51) -- cycle;
	\draw[fill,opacity=0.3,pink] (Z11) -- (Z41) -- (X50) -- cycle;
	
	\draw[fill,opacity=0.3,violet] (X10) -- (Z20) -- (Z50) -- cycle;
	\draw[fill,opacity=0.3,violet] (X10) -- (Z21) -- (Z51) -- cycle;
	\draw[fill,opacity=0.3,violet] (X11) -- (Z20) -- (Z51) -- cycle;
	\draw[fill,opacity=0.3,violet] (X11) -- (Z21) -- (Z50) -- cycle;
	\end{tikzpicture}\captionof{figure}{Bundle diagram for the cluster state on a ring of $n=5$ qubits for contexts $C_1-C_5$.}
\end{figure}
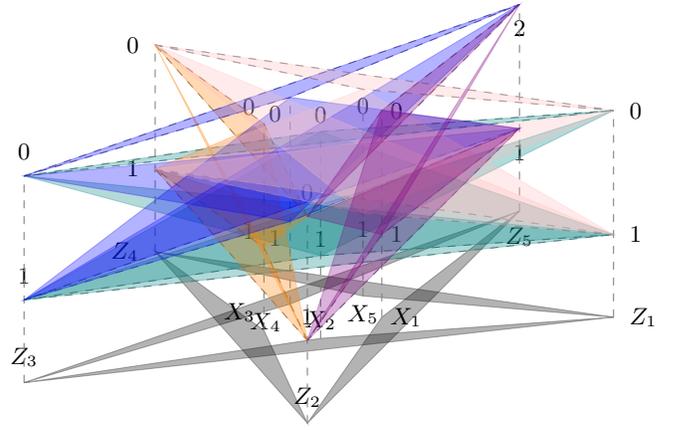
\begin{figure}[H]
	\centering
	\begin{tikzpicture}[scale=1.1]
	\coordinate (X1) at (.7,0,0);
	\node at (X1) [right = .01mm of X1] {$X_1$};
	\coordinate (X2) at (0.21631,0,0.66574);
	\node at (X2) [above = .01mm of X2] {$X_2$};
	\coordinate (X3) at (-0.56631,0,0.41145);
	\node at (X3) [above left = .01mm of X3] {$X_3$};
	\coordinate (X4) at (-0.56631,0,-0.41145);
	\node at (X4) [below left = .01mm of X4] {$X_4$};
	\coordinate (X5) at (0.21631,0,-0.66574);
	\node at (X5) [below = .01mm of X5] {$X_5$};
	\coordinate (Z1) at (3.5,0,0);
	\node at (Z1) [right = 1mm of Z1] {$Z_1$};
	\coordinate (Z2) at (1.0816,0,3.3287);
	\node at (Z2) [above = 1mm of Z2] {$Z_2$};
	\coordinate (Z3) at (-2.8316,0,2.0572);
	\node at (Z3) [above = 1mm of Z3] {$Z_3$};
	\coordinate (Z4) at (-2.8316,0,-2.0572);
	\node at (Z4) [left = 1mm of Z4] {$Z_4$};
	\coordinate (Z5) at (1.0816,0,-3.3287);
	\node at (Z5) [below = 1mm of Z5] {$Z_5$};
\draw[fill,opacity=0.3] (Z1) -- (X2) -- (X4)-- (Z5) -- cycle;
\draw[fill,opacity=0.3] (Z1)--(Z2)--(X3)--(X5) -- cycle;
\draw[fill,opacity=0.3] (X1)--(Z2)--(Z3)--(X4) -- cycle;
\draw[fill,opacity=0.3] (X2)--(Z3)--(Z4)--(X5) -- cycle;
\draw[fill,opacity=0.3] (X1)--(X3)--(Z4)--(Z5)-- cycle;

	\coordinate (X11) at (.7,1,0);
	\node at (X11) [right = .01mm of X11] {$1$};
	\coordinate (X21) at (0.21631,1,0.66574);
	\node at (X21) [above = .01mm of X21] {$1$};
	\coordinate (X31) at (-0.56631,1,0.41145);
	\node at (X31) [above left = .01mm of X31] {$1$};
	\coordinate (X41) at (-0.56631,1,-0.41145);
	\node at (X41) [below left = .01mm of X41] {$1$};
	\coordinate (X51) at (0.21631,1,-0.66574);
	\node at (X51) [below = .01mm of X51] {$1$};
	\coordinate (Z11) at (3.5,1,0);
	\node at (Z11) [right = 1mm of Z11] {$1$};
	\coordinate (Z21) at (1.0816,1,3.3287);
	\node at (Z21) [above = 1mm of Z21] {$1$};
	\coordinate (Z31) at (-2.8316,1,2.0572);
	\node at (Z31) [above = 1mm of Z31] {$1$};
	\coordinate (Z41) at (-2.8316,1,-2.0572);
	\node at (Z41) [left = 1mm of Z41] {$1$};
	\coordinate (Z51) at (1.0816,1,-3.3287);
	\node at (Z51) [below = 1mm of Z51] {$1$};
\draw[dashed, gray] (Z11) -- (X21) -- (X41)-- (Z51) -- cycle;
\draw[dashed, gray] (Z11)--(Z21)--(X31)--(X51) -- cycle;
\draw[dashed, gray] (X11)--(Z21)--(Z31)--(X41) -- cycle;
\draw[dashed, gray] (X21)--(Z31)--(Z41)--(X51) -- cycle;
\draw[dashed, gray] (X11)--(X31)--(Z41)--(Z51)-- cycle;
	\coordinate (X10) at (.7,2.5,0);
	\node at (X10) [right = .01mm of X10] {$0$};
	\coordinate (X20) at (0.21631,2.5,0.66574);
	\node at (X20) [above = .01mm of X20] {$0$};
	\coordinate (X30) at (-0.56631,2.5,0.41145);
	\node at (X30) [above left = .01mm of X30] {$0$};
	\coordinate (X40) at (-0.56631,2.5,-0.41145);
	\node at (X40) [below left = .01mm of X40] {$0$};
	\coordinate (X50) at (0.21631,2.5,-0.66574);
	\node at (X50) [below = .01mm of X50] {$0$};
	\coordinate (Z10) at (3.5,2.5,0);
	\node at (Z10) [right = 1mm of Z10] {$0$};
	\coordinate (Z20) at (1.0816,2.5,3.3287);
	\node at (Z20) [above = 1mm of Z20] {$0$};
	\coordinate (Z30) at (-2.8316,2.5,2.0572);
	\node at (Z30) [above = 1mm of Z30] {$0$};
	\coordinate (Z40) at (-2.8316,2.5,-2.0572);
	\node at (Z40) [left = 1mm of Z40] {$0$};
	\coordinate (Z50) at (1.0816,2.5,-3.3287);
	\node at (Z50) [below = 1mm of Z50] {$0$};
\draw[dashed, gray] (Z10) -- (X20) -- (X40)-- (Z50) -- cycle;
\draw[dashed, gray] (Z10)--(Z20)--(X30)--(X50) -- cycle;
\draw[dashed, gray] (X10)--(Z20)--(Z30)--(X40) -- cycle;
\draw[dashed, gray] (X20)--(Z30)--(Z40)--(X50) -- cycle;
\draw[dashed, gray] (X10)--(X30)--(Z40)--(Z50)-- cycle;
	\draw[dashed, gray] (X1) -- (X10);
	\draw[dashed, gray] (X2) -- (X20);
	\draw[dashed, gray] (X3) -- (X30);	
	\draw[dashed, gray] (X4) -- (X40);	
	\draw[dashed, gray] (X5) -- (X50);
	\draw[dashed, gray] (Z1) -- (Z10);	
	\draw[dashed, gray] (Z2) -- (Z20);	
	\draw[dashed, gray] (Z3) -- (Z30);	
	\draw[dashed, gray] (Z4) -- (Z40);
	\draw[dashed, gray] (Z5) -- (Z50);
	\draw[fill,opacity=0.1,teal] (Z10)--(X10)--(X40)--(Z50)-- cycle;
	\draw[fill,opacity=0.1,teal] (Z10)--(X10)--(X41)--(Z51)-- cycle;
	\draw[fill,opacity=0.1,teal] (Z10)--(X11)--(X40)--(Z51)-- cycle;
	\draw[fill,opacity=0.1,teal] (Z10)--(X11)--(X41)--(Z50)-- cycle;
	\draw[fill,opacity=0.1,teal] (Z11)--(X10)--(X40)--(Z51)-- cycle;
	\draw[fill,opacity=0.1,teal] (Z11)--(X10)--(X41)--(Z50)-- cycle;
	\draw[fill,opacity=0.1,teal] (Z11)--(X11)--(X40)--(Z50)-- cycle;
	\draw[fill,opacity=0.1,teal] (Z11)--(X11)--(X41)--(Z51)-- cycle;

	\draw[fill,opacity=0.1,orange] (Z20)--(X30)--(X50)--(Z10)-- cycle;
	\draw[fill,opacity=0.1,orange] (Z20)--(X30)--(X51)--(Z11)-- cycle;
	\draw[fill,opacity=0.1,orange] (Z20)--(X31)--(X50)--(Z11)-- cycle;
	\draw[fill,opacity=0.1,orange] (Z20)--(X31)--(X51)--(Z10)-- cycle;
	\draw[fill,opacity=0.1,orange] (Z21)--(X30)--(X50)--(Z11)-- cycle;
	\draw[fill,opacity=0.1,orange] (Z21)--(X30)--(X51)--(Z10)-- cycle;
	\draw[fill,opacity=0.1,orange] (Z21)--(X31)--(X50)--(Z10)-- cycle;
	\draw[fill,opacity=0.1,orange] (Z21)--(X31)--(X51)--(Z11)-- cycle;

	\draw[fill,opacity=0.1,blue] (Z30)--(X40)--(X10)--(Z20)-- cycle;
	\draw[fill,opacity=0.1,blue] (Z30)--(X40)--(X11)--(Z21)-- cycle;
	\draw[fill,opacity=0.1,blue] (Z30)--(X41)--(X10)--(Z21)-- cycle;
	\draw[fill,opacity=0.1,blue] (Z30)--(X41)--(X11)--(Z20)-- cycle;
	\draw[fill,opacity=0.1,blue] (Z31)--(X40)--(X10)--(Z21)-- cycle;
	\draw[fill,opacity=0.1,blue] (Z31)--(X40)--(X11)--(Z20)-- cycle;
	\draw[fill,opacity=0.1,blue] (Z31)--(X41)--(X10)--(Z20)-- cycle;
	\draw[fill,opacity=0.1,blue] (Z31)--(X41)--(X11)--(Z21)-- cycle;

	\draw[fill,opacity=0.1,pink] (Z40)--(X50)--(X20)--(Z30)-- cycle;
	\draw[fill,opacity=0.1,pink] (Z40)--(X50)--(X21)--(Z31)-- cycle;
	\draw[fill,opacity=0.1,pink] (Z40)--(X51)--(X20)--(Z31)-- cycle;
	\draw[fill,opacity=0.1,pink] (Z40)--(X51)--(X21)--(Z30)-- cycle;
	\draw[fill,opacity=0.1,pink] (Z41)--(X50)--(X20)--(Z31)-- cycle;
	\draw[fill,opacity=0.1,pink] (Z41)--(X50)--(X21)--(Z30)-- cycle;
	\draw[fill,opacity=0.1,pink] (Z41)--(X51)--(X20)--(Z30)-- cycle;
	\draw[fill,opacity=0.1,pink] (Z41)--(X51)--(X21)--(Z31)-- cycle;

	\draw[fill,opacity=0.1,violet] (Z50)--(X10)--(X30)--(Z40)-- cycle;
	\draw[fill,opacity=0.1,violet] (Z50)--(X10)--(X31)--(Z41)-- cycle;
	\draw[fill,opacity=0.1,violet] (Z50)--(X11)--(X30)--(Z41)-- cycle;
	\draw[fill,opacity=0.1,violet] (Z50)--(X11)--(X31)--(Z40)-- cycle;
	\draw[fill,opacity=0.1,violet] (Z51)--(X10)--(X30)--(Z41)-- cycle;
	\draw[fill,opacity=0.1,violet] (Z51)--(X10)--(X31)--(Z40)-- cycle;
	\draw[fill,opacity=0.1,violet] (Z51)--(X11)--(X30)--(Z40)-- cycle;
	\draw[fill,opacity=0.1,violet] (Z51)--(X11)--(X31)--(Z41)-- cycle;
	\end{tikzpicture}\captionof{figure}{Bundle diagram for the cluster state on a ring of $n=5$ qubits for contexts $C_6-C_{10}$.}
\end{figure}

\onecolumngrid
\begin{center}
\scalebox{1}{
	\begin{tikzpicture}
		[level distance=1cm,
		level 1/.style={sibling distance=3cm},
		level 2/.style={sibling distance=1cm},
		level 3/.style={sibling distance=3cm},
		level 4/.style={sibling distance=1cm},
		triangle/.style = {regular polygon, regular polygon sides=3 },
		square/.style = {regular polygon, regular polygon sides=4 }]
		\tikzstyle{every node}=[font=\small]
		\node {
		\begin{tikzpicture}[triangle/.style = {regular polygon, regular polygon sides=3 },level 1/.style={level distance=2.8cm,sibling distance=0.3cm,every child/.style={edge from parent/.style={draw,white}}},level 2/.style={level distance=1.3cm,sibling distance=0.3cm,every child/.style={edge from parent/.style={draw,black}}},level 3/.style={sibling distance=1cm},grow'=up]
		\node{}child{node{$X_20$}child{node[triangle,rotate=180,fill=teal,opacity=0.3,text width=2mm]{}child {node [rotate=270]{$Z_20*$}}child{node[rotate=270] {$Z_30**$}}}};	
           \end{tikzpicture}
		}
		child {node[triangle,fill=teal,opacity=0.3,right=2cm] { }
			child {node {$Z_11+$}}
			child {node {$Z_31+$}}
		}
		child {node[square,minimum size=.5cm,fill=pink,opacity=0.3,right=2cm] {A}
			child {node {$Z_40$}}
			child {node {$X_50$}
				child {node[square,minimum size=.5cm,fill=orange,opacity=0.3] { }
					child {node {$Z_20$}}
					child {node {$X_30$}
						child {node {$+$}}
					}
					child {node {$Z_10*$}}
					}
				child {node[square,minimum size=.5cm,fill=orange,opacity=0.3] { }
					child {node {$Z_21$}}
					child {node {$X_31$}
						child {node {$+$}}
					}
					child {node {$Z_10*$}}
					}
			}
			child {node {$Z_30**$}}
		}
		child[color=white]{}
		child {node[square,minimum size=.5cm,fill=pink,opacity=0.3] { }
			child {node {$Z_40$}}
			child {node {$X_51$}}
			child {node {$Z_31+$}}
		}
		child {node[square,minimum size=.5cm,fill=pink,opacity=0.3] {}
			child {node {$Z_41$}}
			child {node {$X_50$}}
			child {node {$Z_31+$}}
		}
		child {node[square,minimum size=.5cm,fill=pink,opacity=0.3] {B}
			child {node {$Z_41$}}
			child {node {$X_51$}
				child {node[square,minimum size=.5cm,fill=orange,opacity=0.3] {}
					child {node {$Z_20$}}
					child {node {$X_31$}
						child {node {$+$}}
					}
					child {node {$Z_10*$}}
					}
				child {node[square,minimum size=.5cm,fill=orange,opacity=0.3] { }
					child {node {$Z_21$}}
					child {node {$X_30$}
						child {node {$+$}}
					}
					child {node {$Z_10*$}}
					}
			}
			child {node {$Z_30**$}}
		}
		;
		\end{tikzpicture}
	}

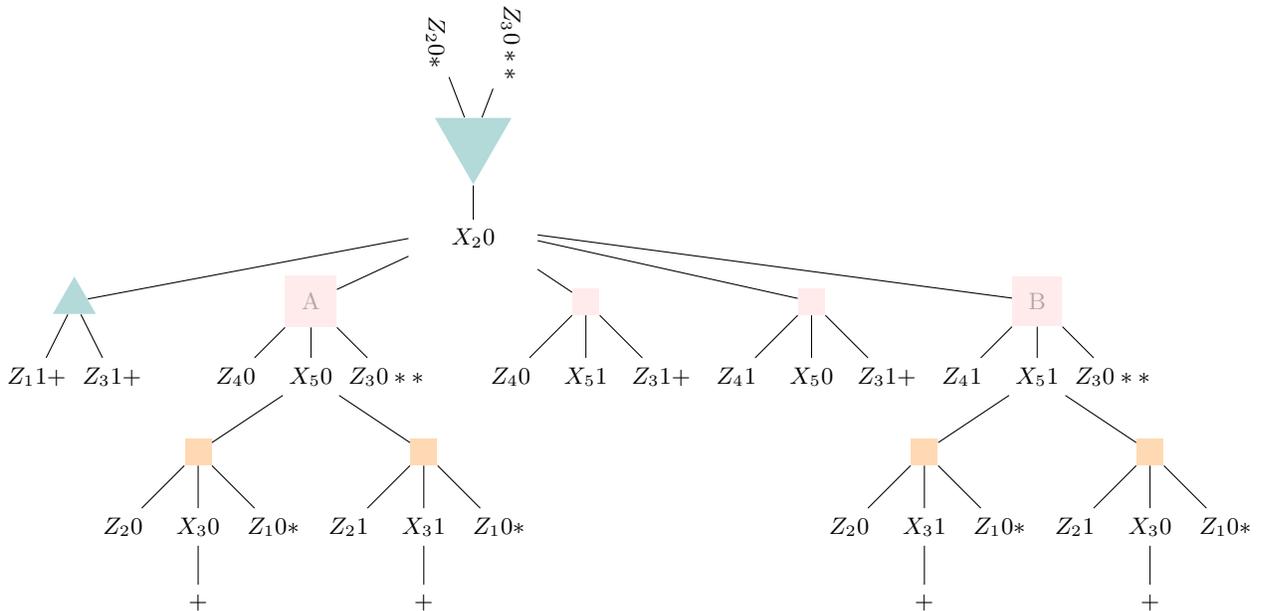
\captionof{figure}{Trying to build a global section with the triangle $Z_10-X_20-Z_30$.}\label{firsttree1}
\end{center}
\twocolumngrid

%% file: 5_Conclusion.tex
\section{Conclusion}
The aim of this work was to exploit \emph{bundle diagrams} to illustrate the contextuality of empirical models involving three or more agents. We considered first bipartite two-agent models with two outcomes as a motivation and then generalized the bundle diagram representation from the two-agent setting to the many agent setting and illustrated the contextuality of a model built on the Greenberger-Horne-Zeilinger state. This representation was then applied to visualise the contextuality of a joint measurability scenario involving a cluster state for a five-qubit ring. There are many interesting questions that arise in depicting contextuality in this way. For example, the bundle diagram itself is an abstract simplicial complex and the contextuality of the model corresponds to the ``twistedness'' of the bundle, like for a M\"obius strip. Finding a way to quantify the connection between orientability and contextuality more precisely is an intriguing research direction.

\begin{acknowledgments}
This work was supported by the the DFG through SFB 1227 (DQ-mat) and the RTG 1991. Helpful correspondence and discussions with Samson Abramsky, Rui Barbosa, and Dan Browne are gratefully acknowledged.
\end{acknowledgments}